\documentclass[reqno]{amsart}

\usepackage[sort&compress,numbers]{natbib}
\usepackage[utf8x]{inputenc}
\usepackage{graphicx,float}
\usepackage{graphicx}
\usepackage{tikz}
\usepackage{pgfplots}
\pgfplotsset{compat=newest}
\usepgfplotslibrary{fillbetween}
\usepackage{amsmath,amsfonts,amssymb,amsthm,mathrsfs}
\usepackage[]{hyperref}
\hypersetup{
  colorlinks   = true,
  urlcolor     = blue,
  linkcolor    = blue,
  citecolor   = red
}

\theoremstyle{plain}
\newtheorem{theorem}{Theorem}
\newtheorem{lemma}[theorem]{Lemma}
\newtheorem{proposition}[theorem]{Proposition}

\newtheoremstyle{note}{\topsep}{\topsep}{\slshape}{}{\scshape}{}{ }{}
\theoremstyle{note}

\newtheorem{remark}[theorem]{Remark}

\newcommand\id{\mathrm{id}}
\newcommand{\ceil}[1]{\left\lceil #1 \right\rceil}
\newcommand{\floor}[1]{\left\lfloor #1 \right\rfloor}
\newcommand{\aspace}{H_{\mathrm{as}}}
\newcommand{\complexes}{\mathbb{C}}
\newcommand{\reals}{\mathbb{R}}
\newcommand{\naturals}{\mathbb{N}}
\newcommand{\ran}[1]{\operatorname{ran}{#1}}

\newcommand\tr{\operatorname{tr}}
\newcommand{\ptr}[2]{\operatorname{tr}_{#1}{#2}}
\newcommand{\proj}[1]{P_{#1}}
\newcommand{\iprod}[2]{\langle #1 , #2 \rangle}
\newcommand{\Span}{\operatorname{span}}

\newcommand\diag{\operatorname{diag}}

\begin{document} 	
\title[New constructive counterexamples to additivity...]{New constructive counterexamples to additivity of minimum output R\'{e}nyi \textit{p}-entropy of quantum channels}

\author{Krzysztof Szczygielski}
\address[K.~Szczygielski]
{Institute of Physics, Faculty of Physics, Astronomy and Informatics, Nicolaus Copernicus University in Toru\'n, Grudzi\c{a}dzka 5/7, 87–100, Toruń, Poland}
\email[K.~Szczygielski]{krzysztof.szczygielski@ug.edu.pl}

\author{Micha{\l} Studzi\'{n}ski}
\address[M. Studzi\'{n}ski]
{Institute of Theoretical Physics and Astrophysics, Faculty of Mathematics, Physics and Informatics, University of Gda\'nsk, Wita Stwosza 57, 80-308 Gda\'nsk, Poland}
\email[M. Studzi\'{n}ski]{michal.studzinski@ug.edu.pl}

\begin{abstract}
In this paper, we present new families of quantum channels for which corresponding minimum output R\'{e}nyi \textit{p}-entropy is not additive. Our manuscript is motivated by the results of Grudka \textit{et al.}, J.~Phys.~A: Math.~Theor. \textbf{43} 425304 and we focus on channels characterized by both extensions and subspaces of the antisymmetric subspace in $\complexes^d \otimes \complexes^d$, which exhibit additivity breaking for $p>2$. 
\end{abstract}

\maketitle

\section{Introduction}
One of key problems considered in quantum information theory is sending information through noisy quantum communication channels. However, contrary to classical channels, their quantum counterpart manifests a phenomenon that the optimal rate to transmit reliably classical/quantum information sometimes is not additive under taking the tensor product of quantum channels. The question of additivity in the above sense is a fundamental one due to the result by Shor \cite{Shor2004}, since it is equivalent to the question of additivity of the Holevo quantity \cite{Holevo} or in other words, whether the product state classical capacity is additive or not. In particular, assuming additivity of the Holevo quantity we can remove its regularization in the definition of the classical capacity of a quantum channel and get a single-letter formula. Also, in the same paper of Shor it was proven that the problem of additivity of the Holevo capacity is equivalent to three other additivity conjectures: additivity of entanglement of formation, the strong super-additivity of entanglement of formation, and finally the \textit{additivity of the minimum output \textit{p}-entropy of a quantum channel}.

For a long time people conjectured that minimum output entropy, or even general R\'{e}nyi \textit{p}-entropies, of quantum channels is additive. These expectations were caused by many of additivity results obtained by certain classes of quantum channels like for example entanglement breaking channels \cite{doi:10.1063/1.1498000}, unital qubit channels \cite{doi:10.1063/1.1500791}, depolarizing channel \cite{1159773}, transpose depolarizing channel \cite{doi:10.1142/S0219749906001633}, and many others \cite{Matsumoto_2004,Alicki2004,Datta_2005,Wolf_2005}. Today we know that the structure of the set of quantum channels is much more complex in this regard. In the general situation we have examples of quantum channels for which minimum output R\'{e}nyi \textit{p}-entropies are not additive for various values of the parameter $p$, including the case when $p\rightarrow 1$ \cite{Hastings2009} reducing the problem the question of additivity for the minimum output von Neumann entropy. However, most of the existing results concern non-constructive proofs, showing additivity violation for a randomly picked channel \cite{Hayden2008,Hastings2009}. Despite the progress in this area only a few explicit constructions for quantum channels violating additivity are known up to now. We can mention here papers by Werner and Holevo \cite{doi:10.1063/1.1498491}, Grudka \textit{et al.} \cite{Grudka_2010}, and Cubitt \textit{et al.} \cite{Cubitt2008}.
The goal of this paper is to extend the known and present new results on explicit construction of additivity breaking for quantum channels. Still we left the most important problem from the perspective of practical application, which is finding constructive counterexample for the von Neumann entropy open. Nevertheless, producing new explicit families of channels with property under investigation certainly gives insight into structure of the set of quantum channels and it is a step further towards mentioned ultimate goal. 

Our approach is motivated by analysis of the maximal Schmidt coefficients of the considered subspaces, which translates to respective bounds of corresponding R\'{e}nyi \textit{p}-entropies. In particular, our contribution is the following:
\begin{enumerate}
    \item Quantum channels by extensions of the antisymmetric subspace. Here as a starting point we use results of Grudka et al. \cite{Grudka_2010} and ask what kind of subspace can be 'added' to antisymmetric one to preserve additivity breaking of quantum channels generated by the considered subspaces. Our construction relies on generalized Bell states \cite{Sych_2009} and gives new classes of additivity breaking quantum channels for any dimension $d$ and $p>2$.
    \item Quantum channels by subspaces of the antisymmetric space. This part is complementary to the first point above. Namely, we consider subspaces of the antisymmetric space still producing quantum channels with considered property of additivity breaking. Our construction works for any $p>2$ and dimensions $d$ greater than a certain transition dimension $d_0=d_0(p)$, for which the expression is delivered.
\end{enumerate}

The structure of the paper is as follows. In Section~\ref{sec:DefOfProblem} we briefly introduce the problem of the additivity breaking of the minimal output R\'{e}nyi \textit{p}-entropies. In Section~\ref{sec:Methodology} we introduce all the most important tools used during further this manuscript. Next, Section~\ref{sec:Results} contains the main results of this paper concerning points 1,2 from the above. The paper is finished with an appendix presenting technical findings used in the main text.

\section{Definition of the problem}
\label{sec:DefOfProblem}

We start with some basic notions considering quantum channels, R\'{e}nyi \textit{p}-entropy and additivity breaking. Let $B_1 (H)$, $B_1 (K)$ be Banach spaces of trace class operators over Hilbert spaces $H$ and $K$, respectively and let $\rho \in B_1(H)$ be a \emph{density operator} of some quantum system, namely $\rho \geqslant 0$ and $\tr{\rho}=1$. Denote also by $D(H)$ a set of all density operators on $H$. By compactness of $\rho$, its spectrum is necessarily pure point, $\sigma (\rho) = \{\lambda_i : i\in \naturals\} \subset [0, 1]$. We define the \emph{R\'{e}nyi \textit{p}-entropy} of $\rho$ via formula
\begin{equation}\label{eq:RenyiEntropyDef}
    S_p (\rho) = \frac{1}{1-p} \log_2 {\tr{\rho^p}} = \frac{1}{1-p} \log_2 {\sum_i \lambda_{i}^{p}},
\end{equation}
for $p \in (0, \infty) \setminus \{1\}$, where the $p$-th power of $\rho$, as well as the last equality, is to be understood via the usual functional calculus. Let now $\mathcal{N} : B_1 (H)\to B_1 (K)$ be a completely positive and trace preserving map, i.e.~a \emph{quantum channel}. The \emph{minimal output \textit{p}-entropy of $\mathcal{N}$} is defined as
\begin{align}
    S^{\mathrm{min}}_{p}(\mathcal{N}) &= \min_{\rho \in D (H)}{S_{p}(\mathcal{N}(\rho))} \\
    &= \min_{h\in H, \| h \| = 1}{S_{p}(\mathcal{N}(\proj{h}))}, \nonumber
\end{align}
where it suffices to minimize exclusively over rank-one orthogonal projections $\proj{h} = \iprod{h}{\cdot} \, h$ in $H$, i.e.~over \emph{pure states}. The minimum output \textit{p}-entropy is \emph{subadditive}, i.e.~for any two quantum channels $\mathcal{N}_1 , \mathcal{N}_2 : B_1 (H) \to B_1 (K)$ we have $S_{p}^{\mathrm{min}}(\mathcal{N}_1 \otimes \mathcal{N}_2) \leqslant S_{p}^{\mathrm{min}}(\mathcal{N}_1) + S_{p}^{\mathrm{min}}(\mathcal{N}_2)$, and, if for some $p$ one has equality for all pairs $(\mathcal{N}_1, \mathcal{N}_2)$ of channels in some dimension, one says that the minimum output \textit{p}-entropy is additive. It was shown already in numerous sources (see introduction) that there exist pairs of channels breaking the additivity condition, i.e.~such that a strict inequality
\begin{equation}
    S_{p}^{\mathrm{min}}(\mathcal{N}_1 \otimes \mathcal{N}_2) < S_{p}^{\mathrm{min}}(\mathcal{N}_1) + S_{p}^{\mathrm{min}}(\mathcal{N}_2)
\end{equation}
is satisfied. In particular, a constructive class of channels was provided in article \cite{Grudka_2010}, where the channels are induced by antisymmetric space $\complexes^d \wedge \complexes^d$.

\section{Methodology and formalism}
\label{sec:Methodology}

\subsection{Background}
 
In this paper we will be focused on the class of channels similar to the one introduced in \cite{Hayden2008,Hastings2009} and further developed in \cite{Grudka_2010}. Let $H \simeq \complexes^m$, $K\simeq \complexes^d$ be finite dimensional Hilbert spaces. We will be considering quantum channels
\begin{equation}
    \mathcal{N}, \overline{\mathcal{N}} : B(H) \to B(K),
\end{equation}
where $m\leqslant d^2$, which action on density matrix can be, via the Stinespring dilation theorem \cite{Stinespring}, represented as
\begin{equation}\label{eq:StinespringNNbar}
    \mathcal{N}(\rho) = \ptr{E}{V\rho V^*}, \quad \overline{\mathcal{N}}(\rho) = \ptr{E}{\overline{V}\rho V^T},
\end{equation}
for some appropriately chosen Hilbert space $E$ and an isometry $V : H \to K \otimes E$; here $\overline{V} = (V^*)^T$ stands for entry-wise complex conjugation of $V$. For simplicity, we put space $E$ isomorphic to $K$, i.e.~$E \simeq \complexes^d$. Since it is known \cite{Aubrun2011} that the minimum output \textit{p}-entropy $S_{p}^{\mathrm{min}}(\mathcal{N})$ does not depend on particular choice of the isometry $V$, only on its range $W = \ran{V}$, of dimension $m$, one can divide a set of all channels from $B(H)$ to $B(K)$ into equivalence classes characterized by subspace $W$. It was shown in \cite{Grudka_2010} that choosing $W$ as (being isomorphic to) an \emph{antisymmetric subspace} $\aspace \subset \complexes^d \otimes \complexes^d$ of dimension $\dim{\aspace} = \binom{d}{2} = \frac{1}{2}d(d-1)$,
 \begin{align}
     W &= \aspace \simeq \complexes^d \wedge \complexes^d \\
     &= \Span{\left\{\frac{1}{\sqrt{2}}(e_i \otimes e_j - e_j \otimes e_i) : i < j\right\}},\nonumber
 \end{align}
where $\{e_i\}_{i=1}^{d}$ is an orthonormal basis spanning $\complexes^d$, produces pair $\mathcal{N},\overline{\mathcal{N}} : B(\complexes^{\binom{d}{2}})\to B(\complexes^d)$ of channels generating additivity breaking for $p>2$. In what follows, we will show that the additivity condition is broken also for differently chosen subspaces $W$, i.e.~we provide prescriptions for constructing additivity breaking channels, beyond results in \cite{Grudka_2010}.
 
For composite channel $\mathcal{N}\otimes\overline{\mathcal{N}}$ we introduce, for computational convenience, distinction between channels $\mathcal{N}$ and $\overline{\mathcal{N}}$ by adding respectively subscripts 1 and 2 to appropriate Hilbert spaces $H$, $K$ and $E$. Likewise, channel $\mathcal{N}\otimes\overline{\mathcal{N}} : B(H_1\otimes H_2) \to B(K_1 \otimes K_2)$ can be put in a Stinespring form
 \begin{equation}
     \mathcal{N}\otimes\overline{\mathcal{N}}(\rho) = \ptr{E_1 E_2}{(V\otimes \overline{V})\rho (V^* \otimes V^T)},
 \end{equation}
 where $V\otimes \overline{V} : H_1\otimes H_2 \to K_1 \otimes E_1 \otimes K_2 \otimes E_2$ is naturally also an isometry. Similarly, one can now identify channel $\mathcal{N}\otimes\overline{\mathcal{N}}$ with a subspace $\mathcal{W} = \ran{(V\otimes \overline{V})} = W_1\otimes W_2$ (with $W_{1,2}$ isomorphic to $W$) of dimension $m^2$ in full tensor product space $K_1 \otimes E_1 \otimes K_2 \otimes E_2$. Our method of showing the additivity breaking is based on the following, well-known theorem (for the reader's convenience, we also supply a short proof of it):
 
 \begin{theorem}\label{thm:TheTheorem}
 Let $\mathcal{N}, \overline{\mathcal{N}} : B(H) \to B(K)$, where $H\simeq\complexes^m$, $K\simeq\complexes^d$, $m\leqslant d^2$, be quantum channels of a form \eqref{eq:StinespringNNbar} and denote again $\mathcal{W} = \ran{ (V\otimes\overline{V})}$. Assume that there exists a constant $C > 0$ such that $S_{p}^{\mathrm{min}}(\mathcal{N}) \geqslant C$. Also, assume that there exists a vector $\psi \in \mathcal{W}$, $\| \psi \| = 1$ and a constant $c>0$ such that $S_{p}(\ptr{E_1 E_2}{\proj{\psi}}) \leqslant c$ for some $p \in (0,\infty ) \setminus \{1\}$. Then, if
 \begin{equation}
     c<2C,
 \end{equation}
 a pair of channels $(\mathcal{N},\overline{\mathcal{N}})$ generates additivity breaking for given $p$.
 \end{theorem}
 
 \begin{proof}
Notice that mapping $\rho \mapsto (V\otimes \overline{V})\rho (V^* \otimes V^T)$ is a completely positive and trace preserving bijection, which uniquely identifies states in $D(H_1\otimes H_2)$ and $D(\mathcal{W})$. Therefore, one immediately notices that, for all $a \in D(\mathcal{W})$ we have
 \begin{equation}
     S_{p}^{\mathrm{min}}(\mathcal{N}\otimes\overline{\mathcal{N}}) \leqslant S_{p}(\ptr{E_1 E_2}{a}).
 \end{equation}
 In particular, we may take $a = \proj{\psi}$ being a rank one projection in $\mathcal{W}$, where $\psi\in\mathcal{W}$ satisfies all the assumptions; then, the minimum output \textit{p}-entropy of composite channel is bounded from above, $S_{p}^{\mathrm{min}}(\mathcal{N}\otimes\overline{\mathcal{N}}) \leqslant c$. Now, since $S_{p}^{\mathrm{min}}(\mathcal{N}) = S_{p}^{\mathrm{min}}(\overline{\mathcal{N}})$, we have
 \begin{align}
     S_{p}^{\mathrm{min}}(\mathcal{N}\otimes\overline{\mathcal{N}}) &\leqslant S_{p}(\ptr{E_1 E_2}{\proj{\psi}}) \\
     &\leqslant c < 2C \nonumber \\
     &\leqslant S_{p}^{\mathrm{min}}(\mathcal{N}) + S_{p}^{\mathrm{min}}(\overline{\mathcal{N}}),\nonumber
 \end{align}
which concludes the proof.
 \end{proof}

It is sometimes more convenient, from conceptual point of view, to look at the full tensor product space in Stinespring representation of channel $\mathcal{N}\otimes\overline{\mathcal{N}}$ as a composite space of four distinct subsystems, $K_1$, $K_2$, $E_1$ and $E_2$ and consider states computed in different \emph{cuts}, so to speak. For this, we introduce two Hilbert tensor product spaces
\begin{align}
    \mathcal{H} &= K_1 \otimes E_1 \otimes K_2 \otimes E_2, \\
    \tilde{\mathcal{H}} &= K_1 \otimes K_2 \otimes E_1 \otimes E_2 ,\nonumber
\end{align}
which are identified by an isometric isomorphism $\eta : \mathcal{H} \to \tilde{\mathcal{H}}$ acting on simple tensors by switching vectors in second and third position, $\eta (k_1 \otimes f_1 \otimes k_2 \otimes f_2) = k_1 \otimes k_2 \otimes f_1 \otimes f_2$. In such case, one can look at space $\mathcal{H}$ as a composite Hilbert space of two subsystems described by spaces $K_1 \otimes E_1$ and $K_2 \otimes E_2$, which may be called the cut $K_1 E_1 : K_2 E_2$. On the other hand, writing $\tilde{\mathcal{H}} = K_{12} \otimes E_{12}$ for $K_{12} = K_1 \otimes K_2$, $E_{12} = E_1 \otimes E_2$ provides different cut $K_1 K_2 : E_1 E_2$. By this distinction, one can easily see the following fact:

\begin{lemma}\label{lemma:MinimalOutEntropyUpperBound}
Let $\xi \in \tilde{\mathcal{H}}$, $\rho = \ptr{E_1 E_2}{\proj{\xi}}$ and $p>1$. Then, R\'{e}nyi \textit{p}-entropy of $\rho$ is bounded from above,
\begin{equation}
    S_{p}(\rho) \leqslant \frac{p}{1-p} \log_2 {\mu_{1}^{2}},
\end{equation}
where $\mu_1$ is the largest coefficient of the Schmidt decomposition of $\xi$ (see also Appendix \ref{appA}).
\end{lemma}

\begin{proof}
Let the Schmidt decomposition of $\xi$ be given by
\begin{equation}
    \xi = \sum_{i} \mu_i u_i \otimes v_i ,
\end{equation}
where $\mu_i \geqslant 0$ and $\{u_i\}$ and $\{v_i\}$ being some orthonormal bases spanning $K_1 \otimes K_2$ and $E_1 \otimes E_2$, respectively. Denote also $\mu_1 = \max{\{\mu_i\}}$. By straightforward algebra, one easily obtains
\begin{equation}\label{eq:rhoReducedcutK12E12}
    \rho = \ptr{E_1 E_2}{\proj{\xi}} = \sum_{i} \mu_{i}^{2} \proj{u_i}.
\end{equation}
Then, since $\sum_i \mu_{i}^{2p} \geqslant \mu_{1}^{2p}$, the correct inequality for $S_p (\rho)$ follows immediately from equation \eqref{eq:RenyiEntropyDef}.
\end{proof}

As the preceding lemma suggests, the largest Schmidt coefficient $\mu_1$ becomes the key quantity for computing the upper bound $c$ of $S_{p}^{\mathrm{min}}(\mathcal{N}\otimes\overline{\mathcal{N}})$. There exists an elegant geometrical bound for $\mu_1$ determined solely by subspaces involved (for proof, see \cite[Lemma 8.7]{Aubrun_Szarek_2017}, or \cite[Lemma 3.3]{Hayden2008}, and Appendix \ref{app:lemmaUpperBoundValueProof}):

\begin{lemma}\label{lemma:UpperBoundValue}
Let $W_1$ and $W_2$ denote subspaces in $K_1 \otimes E_1$ and $K_2 \otimes E_2$, respectively, isomorphic to a certain subspace $W \subset K \otimes E$. Denote $\mathcal{W} = W_1 \otimes W_2 \subset \mathcal{H}$ and let $\psi^+$ be a maximally entangled state between $W_1$ and $W_2$ (i.e.~in cut $K_1 E_1 : K_2 E_2$). Then, the largest coefficient $\mu_1$ of Schmidt decomposition of vector $\eta (\psi^+) \in \tilde{\mathcal{W}}$, i.e.~in cut $K_1 K_2 : E_1 E_2$, satisfies inequality
\begin{equation}
    \mu_{1}^{2} \geqslant \frac{\dim{W}}{\dim{K}\dim{E}} .
\end{equation}
\end{lemma}

\subsection{The method}

\begin{lemma}\label{lemma:LowerBoundValue}
Let $W$ be chosen in such way that the largest Schmidt coefficient $\mu_1$ for every unit vector $\xi \in W$ satisfies inequality $\mu_{1}^{2} \leqslant A$ for some constant $A \in (0,1)$. Then, the lower bound $C$ for minimal output \textit{p}-entropy $S_{p}^{\mathrm{min}} (\mathcal{N})$ can be chosen as
\begin{equation}
    C = \frac{1}{1-p}\log_2 {\left[ (1-A)^p + A^p \right]}.
\end{equation}
\end{lemma}

\begin{proof}
Denote by $D(H)$ and $D(W)$ convex cones of density operators in $B(H)$ and $B(W)$, respectively. Clearly, mapping $\rho \mapsto V \rho V^*$ is a completely positive, trace preserving continuous bijection, therefore we can identify states in these two spaces. This implies that
\begin{align}
    S_{p}^{\mathrm{min}}(\mathcal{N}) &= \min_{h \in H, \| h \|=1}{S_p (\ptr{E}{V\proj{h}V^*})} \\
    &= \min_{w \in W, \| w \| = 1}{S_p (\ptr{E}{\proj{w}})}.\nonumber
\end{align}
Now, let $w\in W$ be given by its Schmidt decomposition, $w = \sum_i \mu_i k_i \otimes f_i$ for $\{k_i\}$, $\{f_i\}$ orthonormal bases in $K$ and $E$, respectively and let $\rho = \ptr{E}{\proj{w}} = \sum_i \mu_{i}^{2} \proj{k_i}$. Let also $\vec{r} = (\mu_{i}^{2}) \in \reals_{+}^{\dim{K}}$ be a vector of its eigenvalues, such that $\rho = \diag{(\mu_{i}^{2})}$. Define also vector $\vec{a}\in\reals_{+}^{\dim{K}}$ by
\begin{equation}
    \vec{a} = (\max{\{A,1-A\}}, \min{\{A,1-A\}}, 0, \, ... \, , \, 0)
\end{equation}
and let $\alpha = \diag{\vec{a}}$ denote a diagonal density matrix of spectrum given by components of $\vec{a}$ (in any order). Then, by the assumed inequality $\mu_{1}^{2} \leqslant A$ we have, after easy algebra, $\vec{r} \preceq \vec{a}$, where $\preceq$ denotes the preorder of \emph{majorization} in $\reals_{+}^{\dim{K}}$ (see Appendix \ref{appMajorization} for exact definition). Since R\'{e}nyi \textit{p}-entropy is a \emph{Schur concave function}, we have $S_{p}(\alpha) \leqslant S_{p}(\rho)$ whenever $\vec{r}\preceq \vec{a}$. However, computing directly, we obtain
\begin{align}
    S_{p}(\alpha) &= \frac{1}{1-p} \log_{2}{\tr{\alpha^p}} \\
    &= \frac{1}{1-p}\log_{2}\left[ (1-A)^p + A^p \right] = C.\nonumber
\end{align}
Then $S_{p}(\ptr{E}{\proj{w}}) \geqslant C$ for every $w\in W$, and so the minimal output \textit{p}-entropy is also bounded by $C$.
\end{proof}

By combining all the lemmas outlined above we arrive at the following result, which may be understood as a sufficient condition -- and a practical \emph{test} -- for additivity breaking for channels defined by various subspaces $W$ in $K \otimes E$, suitable in the regime $p > 1$ (we will elaborate on this limitation). The following theorem will be directly applied in the next section.

\begin{theorem}\label{thm:TheTheorem2}
Let $\mathcal{N}, \overline{\mathcal{N}} : B(H) \to B(K)$, where $H\simeq\complexes^m$, $K\simeq\complexes^d$, $m\leqslant d^2$, be quantum channels of a form \eqref{eq:StinespringNNbar} identified by a subspace $W\subset K \otimes E$. Assume the largest Schmidt coefficient $\mu_1$ of any vector $\xi \in W$, $\| \xi \| = 1$, satisfies inequality $\mu_{1}^{2} \leqslant A$ for some $A \in (0,1)$. The following statements hold:
\begin{enumerate}
    \item \label{item:TheTheorem2item1} If it holds that
\begin{equation}\label{eq:BreakingInequality}
    \frac{\dim{W}}{\dim{K}\dim{E}} > \left[ (1-A)^p + A^p \right]^{\frac{2}{p}},
\end{equation}
then a pair of channels $(\mathcal{N},\overline{\mathcal{N}})$ generates additivity breaking of minimum R\'{e}nyi \textit{p}-entropy for given $p > 1$.
\item \label{item:TheTheorem2item2} If it holds that
\begin{align}\label{eq:Limitation}
    \frac{\dim{W}}{\dim{K}\dim{E}} &> \max{\{A^2 , (1-A)^2\}} \\
    &= \begin{cases} A^2, & A \geqslant \frac{1}{2}, \\ (1-A)^2, & A < \frac{1}{2}, \end{cases} \nonumber
\end{align}
then there exists some $p_0 > 1$ such that a pair of channels $(\mathcal{N},\overline{\mathcal{N}})$ generates additivity breaking of minimum R\'{e}nyi \textit{p}-entropy for all $p > p_0$. Parameter $p_0$ may be further roughly estimated by
\begin{equation}\label{eq:Theorem2p0Bound}
    p_0 \geqslant \left( 1 + \log_{4}{\frac{\dim{W}}{\dim{K}\dim{E}}}\right)^{-1}.
\end{equation}
\end{enumerate}
\end{theorem}

\begin{proof}
Ad \ref{item:TheTheorem2item1}. Assume $p > 1$ and denote by $\psi^+$ a maximally entangled state in space $\mathcal{W} = W_1 \otimes W_2$. Then, the largest Schmidt coefficient of vector $\eta(\psi^+) \in \tilde{\mathcal{W}}$, i.e.~in cut $K_1 K_2 : E_1 E_2$, satisfies
\begin{equation}
    \mu_{1}^{2} \geqslant \frac{\dim{W}}{\dim{K}\dim{E}},
\end{equation}
which comes from lemma \ref{lemma:UpperBoundValue}. In result, the R\'{e}nyi \textit{p}-entropy of state $\rho = \ptr{E_1 E_2}{\proj{\eta(\psi^+)}}$ is, by taking $\xi = \eta (\psi^+)$ in lemma \ref{lemma:MinimalOutEntropyUpperBound}, bounded from above,
\begin{equation}
    S_{p}(\rho) \leqslant \frac{p}{1-p} \log_{2}{\frac{\dim{W}}{\dim{K}\dim{E}}} = c.
\end{equation}
However, after choosing $\psi = \psi^+$ in theorem \ref{thm:TheTheorem}, one concludes that $c$ is also the upper bound of $S_{p}^{\mathrm{min}}(\mathcal{N} \otimes\overline{\mathcal{N}})$. On the other hand, if largest Schmidt coefficient $\mu_1$ of any unit vector $\xi \in W$ satisfies $\mu_{1}^{2} \leqslant A$, then by lemma \ref{lemma:LowerBoundValue} we have the lower bound for minimal output \textit{p}-entropy of channel $\mathcal{N}$ defined by $W$,
\begin{equation}
    S_{p}^{\mathrm{min}}(\mathcal{N}) \geqslant \frac{1}{1-p}\log_2 {\left[ (1-A)^p + A^p \right]} = C.
\end{equation}
The additivity breaking condition $c < 2C$ provided by theorem \ref{thm:TheTheorem} can be then shown to be equivalent to claimed inequality \eqref{eq:BreakingInequality} after simple algebra involving monotonicity of logarithm and under assumption $p > 1$.

Ad \ref{item:TheTheorem2item2}. Let $u = \frac{\dim{W}}{\dim{K}\dim{E}}$ for brevity. The inequality \ref{eq:BreakingInequality} may be conveniently rewritten as
\begin{equation}\label{eq:Limitation2}
    u > \| (A, 1-A) \|_{p}^{2}
\end{equation}
for $\| \cdot \|_p$ being the usual $l^p$ norm, here invoked in $\reals^2$. Let us define an open region $\mathcal{S}_p \subset [0,1]^2$ of real plane via condition
\begin{equation}
    \mathcal{S}_p = \{ (A, u) : \| (A, 1-A) \|_{p}^{2} < u < 1\}.
\end{equation}
Every $\mathcal{S}_p$ is then bounded from below by a curve $\kappa_p (A) = \| (A, 1-A) \|_{p}^{2}$. Since $l^p$ norm is monotone with $p$, $\|\cdot\|_\infty \leqslant \| \cdot \|_p \leqslant \| \cdot \|_q$ whenever $q < p$, we infer $\mathcal{S}_q \subset \mathcal{S}_p$ for all $q < p$ i.e.~curve $\kappa_p$ lays below curve $\kappa_q$. For $p = 1$ we have $\| a \|_{1}^{2} = (A + 1 - A)^2 = 1$ so the inequality \eqref{eq:Limitation2} cannot be satisfied (otherwise $W$ could not be a subspace in $K\otimes E$) and $\mathcal{S}_1 = \emptyset$. For the same reason, all sets $\mathcal{S}_p$ for $p < 1$ are empty as well and so a net $(\mathcal{S}_p)_{p > 1}$ is increasing (in sense of inclusions). Since all $l^p$ norms are bounded from below by $\|\cdot\|_\infty$, the supremum norm, we see $(\mathcal{S}_p)_{p>1}$ is also upper bounded by a \emph{maximal} set
\begin{equation}
    \mathcal{S}_\infty = \bigcup_{p>1} \mathcal{S}_p = \{(A,u) : u > \| (A, 1-A) \|_{\infty}^{2}\}.
\end{equation}
Naturally $\| (A, 1-A)\|_{\infty} = \max{\{A, 1-A\}}$, the \emph{maximum norm}, so \eqref{eq:Limitation} comes by simple algebra. Since $(\mathcal{S}_p)_{p>1}$ is increasing, define, for given point $(A,u)$,
\begin{equation}
    \mathcal{S}_{p_0} = \inf\{\mathcal{S}_p : (A,u)\in S_p\},
\end{equation}
in sense of inclusions. Such set is bounded from below by $\kappa_{p_0}$ and $\mathcal{S}_{p_0} \subset \mathcal{S}_{p}$ for all $p > p_0$. Then, while $(A,u) \notin \mathcal{S}_{p_0}$, we still have $(A,u) \in \mathcal{S}_{p}$ for $p > p_0$, i.e.~point $(A,u)$ describes a subspace in $K\otimes E$ which breaks additivity for all $p>p_0$. Finally, a claimed, rough estimation for lower bound of $p_0$ may be found as follows. Choose any $u$ and let $p_* > 1$ be a parameter specifying a curve $\kappa_{p_*}$ containing a point $(\frac{1}{2},u)$. Then, any other point $(A, u)$ with $A\neq\frac{1}{2}$ will necessarily lay on some curve $\kappa_p$ being more convex than $\kappa_{p_*}$, i.e.~specified by some $p > p_*$. In other words, $u$ and $p_*$ satisfy
\begin{equation}
    u = \left\| \left(\frac{1}{2}, 1-\frac{1}{2}\right) \right\|_{p_*}^{2} = 4^{\frac{1}{p_{*}}-1},
\end{equation}
which then yields
\begin{equation}
    p_* = \frac{1}{1 + \log_{4}{u}},
\end{equation}
which is the claim. This finalizes the proof.
\end{proof}

In fig. \ref{fig:BreakingRegion} we present a visualization of different regions $\mathcal{S}_p$ together with bounding curves $\kappa_p$ for various values $p > 1$. Each point $(A,u)$, where $u = \frac{\dim{W}}{\dim{K}\dim{E}}$, represents some class of subspaces in $K\otimes E$ which share the same upper bound $A$ for square of the largest Schmidt coefficient of unit vector. Note that in the case of antisymmetric space $\aspace \simeq \complexes^d \wedge \complexes^d$ we have $\dim{\aspace} = \binom{d}{2}$ and formula \eqref{eq:Theorem2p0Bound} of theorem \ref{thm:TheTheorem2} yields the lower bound for $p_0$ to be
\begin{equation}
    p_0 \geqslant \left(1 + \log_{4}\frac{d-1}{2d}\right)^{-1} > \left(1 + \log_{4}\frac{1}{2}\right)^{-1} > 2
\end{equation}
after simple algebra. In result, for every $d \geqslant 2$ one finds a set $\mathcal{S}_{p_0}$ for some $p_0 > 2$  such that $(\frac{1}{2},u) \in \bigcap_{p > p_0} \mathcal{S}_p$, i.e.~antisymmetric space generates additivity breaking for $p$, $d$ sufficiently large and we recover the famous result by Grudka \textit{et al.} \cite{Grudka_2010}.

\begin{figure}[h!]
    \centering
    \begin{tikzpicture}[
    declare function={
    func(\x) = (\x<=0.5) * (1 - \x) * (1 - \x)  +
     (\x>0.5) * (\x*\x);
    }]
    
    \begin{axis}[
        ylabel = {$u = \frac{\dim{W}}{\dim{K}\dim{E}}$},
        ylabel style={rotate=-90},
        xlabel = {A},
        ymax=1, xmax=1,
        ymin=0, xmin=0,
        xtick = {0, 1},
        ytick = {0, 1},
        extra x ticks = {0.5}, extra x tick labels = {$\frac{1}{2}$},
        extra y ticks = {0.25}, extra y tick labels = {$\frac{1}{4}$},
        extra tick style={grid=major},
        unit vector ratio = 1 1]

        \path[name path=axis] (axis cs:0,1) -- (axis cs:1,1);
        
        \addplot[name path=f1, domain=0:1, thick, color=black]{func(x)};
        \addplot[thick, color=blue, fill=blue, fill opacity=0.05]
            fill between[
                of=axis and f1,
                soft clip={domain=0:1}];
            ];
        \node[black] at (0.67,0.28) {$p=\infty$};            
                
        \addplot[name path=f2, domain=0:1, thick, color=blue, dashed]{((1 - x)^5 + x^5)^(2/5)};
        \addplot[thick, color=blue, fill=blue, fill opacity=0.1]
            fill between[
                of=axis and f2,
                soft clip={domain=0:1}];
            ];
        \node[blue] at (0.5,0.4) {$p=5$};
        
        \addplot[name path=f3, domain=0:1, thick, color=red, dashed]{((1 - x)^2 + x^2)};
        \addplot[thick, color=blue, fill=blue, fill opacity=0.1]
            fill between[
                of=axis and f3,
                soft clip={domain=0:1}];
            ];
        \node[red] at (0.5,0.57) {$p=2$};
        
        \addplot[name path=f4, domain=0:1, thick, color=blue, dashed]{((1 - x)^1.4 + x^1.4)^(2/1.4)};
        \addplot[thick, color=blue, fill=blue, fill opacity=0.1]
            fill between[
                of=axis and f4,
                soft clip={domain=0:1}];
            ];
        \node[blue] at (0.5,0.75) {$p=1.4$};
        
        \addplot[name path=f5, domain=0:1, thick, color=blue, dashed]{((1 - x)^1.1 + x^1.1)^(2/1.1)};
        \addplot[thick, color=blue, fill=blue, fill opacity=0.1]
            fill between[
                of=axis and f5,
                soft clip={domain=0:1}];
            ];
        \node[blue] at (0.5,0.93) {$p=1.1$};
    \end{axis}
    \end{tikzpicture}
    \caption{Visualization of different regions $\mathcal{S}_p$ and their bounding curves $\kappa_p$ for $p = 1.1$, $1.4$, $2.0$, $5.0$ and $\infty$ (top to bottom). Curves $\kappa_2$ and $\kappa_\infty$ are marked in red and black, respectively.}
    \label{fig:BreakingRegion}
\end{figure}
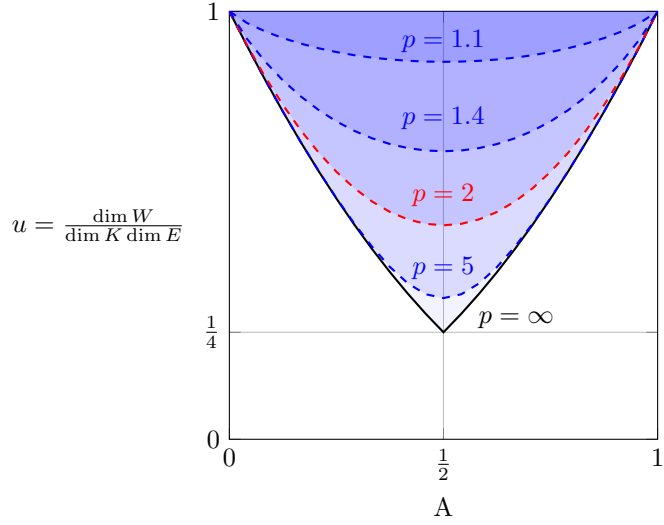

\begin{remark}
Note that the size of region $\mathcal{S}_p$ decreases when $p\to 1^+$ and vanishes completely for $p = 1$. This means that it becomes relatively more and more difficult to detect additivity breaking for lower values of $p$, while it is even impossible to achieve so in case $p=1$. This observation ultimately limits robustness of methods based on estimations of bounds for largest Schmidt coefficients, as outlined in this whole section, to cases of R\'{e}nyi \textit{p}-entropies for $p$ strictly larger than 1. Naturally, the case of Shannon (or, equivalently, von Neumann) entropy, i.e.~case $p=1$, lays outside the scope of theorem \ref{thm:TheTheorem2} as well. Subspaces determined by value of $A$ close to $\frac{1}{2}$ are expected to exhibit the largest potential for additivity breaking (with the antisymmetric subspace $\aspace$ being the most prominent example).
\end{remark}

\section{Results}
\label{sec:Results}

In this section, we present a number of different classes of additivity breaking channels, each one characterized by different choice of subspace $W$. The key point in every case will be to find a "sufficiently well" developed bounds $c$ and $C$, as stated in theorem \ref{thm:TheTheorem}. For sake of brevity, we will be denoting $d = \dim{K} = \dim{E}$ from now on.

It is a known fact that a square of the largest Schmidt coefficient of any vector of unit norm in antisymmetric space $\aspace \subset K \otimes E$ is bounded from above by $\frac{1}{2}$  (see Prop. 1 in \cite{Grudka_2010}). Below, we will present, for reader's convenience, a short proof of this statement. Let $P_{\mathrm{as}} : K\otimes E \to \aspace$ be an orthogonal projection onto antisymmetric subspace. One checks directly that $P_{\mathrm{as}}$ may be represented as
\begin{equation}\label{eq:PasDecomposition}
    P_{\mathrm{as}} = \frac{1}{2} (I - V),
\end{equation}
where $I$ is the identity and $V : K\otimes E \to K \otimes E$ is an operator, which swaps vector order in simple tensors, i.e.~it acts as $V(h \otimes f) = f \otimes h$ (for isomorphic spaces $K$ and $E$).
\begin{lemma}\label{lemma:PasInequality}
Projection $P_{\mathrm{as}} : K\otimes E \to \aspace$ satisfies
\begin{equation}
    \sup_{\| x \otimes y\| = 1}{\| P_{\mathrm{as}}(x\otimes y) \|^{2}} = \frac{1}{2}.
\end{equation}
In result, the largest Schmidt coefficient $\mu_{1}$ of any unit vector in $\aspace$ satisfies $\mu_{1}^{2} \leqslant \frac{1}{2}$.
\end{lemma}

\begin{proof}
Applying \eqref{eq:PasDecomposition} we have, after easy algebra,
\begin{align}
    &\sup_{\| x \otimes y\| = 1}{\| P_{\mathrm{as}}(x\otimes y) \|^{2}} \\
    &= \sup_{\|x\|,\|y\|=1}{\langle x\otimes y , \frac{1}{2}(I-V)(x\otimes y) \rangle} \nonumber \\
    &= \frac{1}{2} \sup_{\|x\otimes y\|=1}{\left(\|x\otimes y\|^{2} - |\langle x , y \rangle |^{2}\right)} \nonumber \\
    &= \frac{1}{2} \left( 1 - \inf_{\|x\otimes y\|=1}{|\langle x,y \rangle|^{2}}\right) \nonumber
\end{align}
which comes by continuity of mapping $(x,y)\mapsto |\langle x , y \rangle|^{2}$. This however is immediately found to be $\frac{1}{2}$, since taking any two mutually orthogonal vectors $x$, $y$ automatically yields the infimum in above expression to be 0. Then $\mu_{1}^{2} \leqslant \frac{1}{2}$ comes from lemma \ref{lemma:SchmidtCoeffOfProjection} from Appendix~\ref{appA}.
\end{proof}

\subsection{Extensions of antisymmetric subspace}
\label{Sec:extensions}

Informally speaking, the key idea here is to extend the antisymmetric subspace $\aspace \simeq \complexes^d \wedge \complexes^d$ in a way, which does not modify the value of maximal Schmidt coefficient too much. In particular, we seek for a space
\begin{equation}
    W = \aspace \oplus X,
\end{equation}
where $W$ is still a subspace in $K\otimes E$ and $X$ is of dimension $n \in \{1, \, \ldots \, , \, \frac{1}{2}d(d+1)\}$, perhaps for $d$ large enough. Let then $X = \Span{\{\phi_i\}}$, where $\{\phi_i\}$ stands for some orthonormal system in $K\otimes E$ of vectors orthogonal to $\aspace$ and introduce orthogonal projection operators $P_W$, $P_{\mathrm{as}}$ and $P_X$ onto $W$, $\aspace$ and $X$, such that $P_W = P_{\mathrm{as}} + P_X$. Observe that, by lemmas \ref{lemma:SchmidtMaxAsSupremum}, \ref{lemma:SchmidtCoeffOfProjection} and \ref{lemma:SchmidtCoeffOfDirectSum}, we have
\begin{align}\label{eq:MuMaxSqinequality}
    \mu_{1}^{2} &\leqslant \sup_{\|x\otimes y\|=1}{\| P_W (x\otimes y) \|^{2}} \\
    &\leqslant \sup_{\|x\otimes y\|=1}{\| P_{\mathrm{as}} (x\otimes y) \|^{2}} + \sup_{\|x\otimes y\|=1}{\| P_{X} (x\otimes y) \|^{2}} \nonumber \\
    &\leqslant \frac{1}{2} + \sup_{\|x\otimes y\|=1}{\sum_{i=1}^{n}|\langle \phi_i , x\otimes y \rangle|^2} \nonumber \\
    &\leqslant \frac{1}{2} + \sum_{i=1}^{n} \mu_{1}^{2}(\phi_i), \nonumber
\end{align}
where $\mu_{1}^{2}(\phi_i)$ is the maximal Schmidt coefficient of basis vector $\phi_i$. This means that in order to minimize the impact of adding space $X$ to $\aspace$ on the upper bound of maximal Schmidt coefficient of $\xi \in W$ (and on the lower bound $C$ for $S_{p}^{\mathrm{min}}(\mathcal{N})$), one should seek for basis $\{\phi_i\}$ of possibly smallest value of $\mu_{1}^{2}(\phi_i)$. A natural candidate for such basis could be a one given exclusively by maximally entangled vectors, since it is well known that a maximally entangled unit vector provides a lowest possible value of $\mu_{1}^{2}$, equal to reciprocal of a dimension. For this, we construct an extending space $X$ as a linear span of maximally entangled vectors.
\vskip\baselineskip
In \cite{Bennett1993}, Bennet \textit{et al.} introduced a very robust generalization of Bell states, suitable for description of $d$-dimensional systems. Let $\{h_i\}_{i=0}^{d-1}\subset K$ and $\{f_i\}_{i=0}^{d-1} \subset E$ be some orthonormal computational bases. We define $d^2$ mutually orthonormal maximally entangled vectors
\begin{equation}\label{eq:BellState}
    \psi_{r,s} = \frac{1}{\sqrt{d}} \sum_{j=0}^{d-1}\omega^{rj} h_j \otimes f_{(j+s)\,\mathrm{mod}\,d},
\end{equation}
where $\omega = e^{\frac{2\pi i}{d}}$ and $r,s\in\{0,\, \ldots \, , \, d-1\}$, which span $K\otimes E$ (we use a numbering convention from \cite{Bennett1993}). We show the following

\begin{lemma}\label{lemma:ExistenceOfBellStates}
For each dimension $d \geqslant 2$ there exist 
\begin{equation}\label{eq:sOfDBellStates}
    s(d) = \frac{d}{2} \left[2 + (d+1) \, \mathrm{mod} \, 2\right]
\end{equation}
generalized Bell states $\psi_{r,s}$ in space $\aspace^{\perp}$. Explicitly, we have $s(d) = d$ if $d$ is odd and $s(d) = \frac{3}{2}d$ if $d$ is even.
\end{lemma}

\begin{proof}
Note that every vector $x = \sum_{ij} x_{ij} h_i \otimes f_j \in K\otimes E$ is uniquely identified with matrix $[x_{ij}]$, so it is enough to analyze appropriate matrices instead. Likewise, the full tensor product $K\otimes E$ is isomorphic, as a Hilbert space, with $M_{d}(\complexes)$ equipped with Frobenius (Hilbert-Schmidt) inner product. Then, it is easy to notice that space $\aspace$ is identified with a subspace of all \emph{antisymmetric matrices}, i.e.~matrices $m \in M_{d}(\complexes)$ satisfying $m^T = -m$, and its orthogonal complement $\aspace^{\perp}$ is isomorphic to subspace of all \emph{symmetric matrices}, i.e.~such that $m^T = m$. Then, question of existence of generalized Bell states in $\aspace^{\perp}$ can be rephrased as a problem of finding Bell states described by symmetric matrices; let us call such states simply \emph{symmetric} for brevity. Fix $d \geqslant 2$ and $r \in \{0,\,  \ldots \, , \, d-1\}$. One checks that vector $\psi_{r,0}$ is \emph{diagonal}, i.e.~is  represented by a diagonal matrix
\begin{equation}
    \hat{\psi}_{r,0} = \frac{1}{\sqrt{d}}\diag\left\{1,\, \omega^{r}, \, \omega^{2r}, \, \ldots \, , \, \omega^{(d-1)r}\right\},
\end{equation}
which is clearly symmetric, so we know that there exists at least one Bell state in $H_{\mathrm{as}}^{\perp}$ for each $r$. Prescription \eqref{eq:BellState} allows for easy generation of all vectors $\psi_{r,s}$ for $s > 0$. Below we present a short textual description of this procedure, accompanied by an explicit example in case $d=6$ for reader's convenience. Having matrix $\hat{\psi}_{r,s}$, matrix $\hat{\psi}_{r,s+1}$ is built by cyclically pushing all columns of $\hat{\psi}_{r,s}$ one position to the right (with the last column of $\hat{\psi}_{r,s}$ becoming the first one in $\hat{\psi}_{r,s+1}$). In result, matrix $\hat{\psi}_{r,1}$ will have only one non-zero element $\frac{1}{\sqrt{d}}\omega^{(d-1)r}$ in  position $(d,1)$; likewise, matrix $\hat{\psi}_{r,2}$ will contain $\frac{1}{\sqrt{d}}\omega^{(d-1)r}$ in position $(d,2)$ and $\frac{1}{\sqrt{d}}\omega^{(d-2)r}$ in position $(d-1,1)$ and so forth. In $s$-th step of this process, all columns of initial matrix $\hat{\psi}_{r,0}$ are pushed by $s$ positions to the right and the only non-zero element of the first column of resulting matrix $\hat{\psi}_{r,s}$ lays in $(d-s+1)$-th row. Now, if $\hat{\psi}_{r,s}$ is expected to be non-diagonal and symmetric, number of leading zeros in its first column must equate number of leading zeros in its first row, so we necessarily need $s = d-s$ or $s = \frac{d}{2}$. This is however possible only for even $d$, so we conclude that if dimension $d$ is odd, we only have $s(d) = d$ Bell states in $\aspace^\perp$, precisely \emph{diagonal} ones, $\psi_{0,0}$, $\psi_{1,0}$, \ldots , $\psi_{d-1,0}$. Assume then $d$ is even. In such case, matrix $\hat{\psi}_{r,0}$ may be put in a block form,
\begin{equation}\label{eq:psir0}
    \hat{\psi}_{r,0} = \frac{1}{\sqrt{d}} \left[ \begin{array}{c|c} \mathbf{M}_{11} & \mathbf{0} \\ \hline \mathbf{0} & \mathbf{M}_{22} \end{array} \right],
\end{equation}
where
\begin{align}
    \mathbf{M}_{11} &= \diag\left\{1,\, \omega^{r}, \, \ldots \, , \, \omega^{\left( \frac{d}{2}-1\right) r}\right\},\\
    \mathbf{M}_{22} &= \diag\left\{\omega^{\frac{d}{2}r},\, \omega^{\left( \frac{d}{2}+1\right) r}, \, \ldots \, , \, \omega^{(d-1) r}\right\}.\nonumber
\end{align}
The only candidate, for given $r$, for a symmetric and non-diagonal Bell state is the one for $s = \frac{d}{2}$, so we must have $\hat{\psi}_{r,d/2} = \hat{\psi}_{r,d/2}^{T}$. However, as $\hat{\psi}_{r,d/2}$ was constructed by shifting columns to the right in cyclic manner, one checks that $\hat{\psi}_{r,d/2}$ is simply the block matrix \eqref{eq:psir0} with switched columns,
\begin{equation}
    \hat{\psi}_{r,d/2} = \frac{1}{\sqrt{d}} \left[ \begin{array}{c|c} \mathbf{0} & \mathbf{M}_{11} \\ \hline \mathbf{M}_{22} & \mathbf{0} \end{array} \right].
\end{equation}
The symmetry condition yields then $\mathbf{M}_{11} = \mathbf{M}_{22}$, or, explicitly,
\begin{align}\label{eq:SymmRelation}
    \omega^{rk} = \omega^{r\left(\frac{d}{2} + k\right)} \quad &\text{for all } k \in \{0,\, \ldots \, , \, d-1\},\\
    &\text{or} \quad \omega^{\frac{rd}{2}} = e^{i\pi r} = 1.\nonumber
\end{align}
However, satisfying $e^{i\pi r} = 1$ as appeared in \eqref{eq:SymmRelation} is possible only for \emph{even} $r$. In consequence, a total number of non-diagonal Bell states in $\aspace^{\perp}$ for even values of $d$ equates a number of all even naturals $r \in \{0, \, \ldots \, , \, d-1\}$ which is $\frac{d}{2}$. Together with $d$ diagonal vectors $\psi_{r,0}$ we have $s(d) = \frac{3}{2}d$ symmetric Bell states. Summarizing,
\begin{equation}
    s(d) = \begin{cases}
        d, & d \text{ odd,} \\
        \frac{3}{2}d, & d \text{ even,}
        \end{cases}
\end{equation}
which can be further checked to be equivalent to claimed formula \eqref{eq:sOfDBellStates}. As an example, below we demonstrate the aforementioned procedure in case $d=6$ (with dots representing zeros):
\small
\begin{align}
\hat{\psi}_{r,0} &= \frac{1}{\sqrt{6}}
    \left[
    \begin{array}{ccc|ccc}
        1 & \cdot    & \cdot & \cdot & \cdot & \cdot \\
        \cdot   & \omega^r & \cdot & \cdot & \cdot & \cdot \\
        \cdot   & \cdot    & \omega^{2r} & \cdot & \cdot & \cdot \\
        \hline
        \cdot   & \cdot    & \cdot & \omega^{3r} & \cdot & \cdot \\
        \cdot   & \cdot    & \cdot & \cdot & \omega^{4r} & \cdot \\
        \cdot   & \cdot    & \cdot & \cdot & \cdot & \omega^{5r}
    \end{array} \right] \\
    \xrightarrow{s+1} \hat{\psi}_{r,1} &= \frac{1}{\sqrt{6}} \left[
    \begin{array}{cccccc}
        \cdot & 1 & \cdot & \cdot & \cdot & \cdot \\
        \cdot & \cdot & \omega^r & \cdot & \cdot & \cdot \\
        \cdot & \cdot & \cdot & \omega^{2r} & \cdot & \cdot \\
        \cdot & \cdot & \cdot & \cdot & \omega^{3r} & \cdot \\
        \cdot & \cdot & \cdot & \cdot & \cdot & \omega^{4r} \\
        \omega^{5r} & \cdot & \cdot & \cdot & \cdot & \cdot
    \end{array} \right] \nonumber \\
    \xrightarrow{s+1}\hat{\psi}_{r,2} &= \frac{1}{\sqrt{6}} \left[
    \begin{array}{cccccc}
        \cdot & \cdot & 1 & \cdot & \cdot & \cdot \\
        \cdot & \cdot & \cdot & \omega^r & \cdot & \cdot \\
        \cdot & \cdot & \cdot & \cdot & \omega^{2r} & \cdot \\
        \cdot & \cdot & \cdot & \cdot & \cdot & \omega^{3r} \\
        \omega^{4r} & \cdot & \cdot & \cdot & \cdot & \cdot \\
        \cdot & \omega^{5r} & \cdot & \cdot & \cdot & \cdot
    \end{array} \right] \nonumber \\
    \xrightarrow{s+1}\hat{\psi}_{r,3} &= \frac{1}{\sqrt{6}} \left[
    \begin{array}{ccc|ccc}
        \cdot & \cdot & \cdot & 1 & \cdot & \cdot \\
        \cdot & \cdot & \cdot & \cdot & \omega^r & \cdot \\
        \cdot & \cdot & \cdot & \cdot & \cdot & \omega^{2r} \\
        \hline
        \omega^{3r} & \cdot & \cdot & \cdot & \cdot & \cdot \\
        \cdot & \omega^{4r} & \cdot & \cdot & \cdot & \cdot \\
        \cdot & \cdot & \omega^{5r} & \cdot & \cdot & \cdot
    \end{array} \right]\nonumber
\end{align}
\normalsize
Then, simple analysis shows that the relation \eqref{eq:SymmRelation} for such elements must be met. 
\end{proof}

Lemma \ref{lemma:ExistenceOfBellStates} assures that we always have a sufficient number of symmetric Bell states in our disposal in order to construct appropriate extending spaces, regardless of $d$. Below we formulate our main result of this section:

\begin{theorem}\label{thm:ResultBell} The following statements hold:
\begin{enumerate}
    \item\label{item:ResultBellitem1} 
    For every $d > 6$ there exists a sequence of extending subspaces $(X_n)_{n=1}^{n_0}$, where $\dim{X_n}=n$ and $n_0 < \ceil{\frac{d}{2}}$, each one orthogonal to $\aspace$, such that each subspace $W_n = \aspace\oplus X_n$ defines a pair of quantum channels $\mathcal{N}_n, \overline{\mathcal{N}}_n  : B(H_n) \to B(K)$, $H_n\simeq W_n$, $K \simeq \complexes^d$, providing additivity breaking of minimal R\'{e}nyi \textit{p}-entropy for all $p$ greater then some $p_0 > 2$.
    \item\label{item:ResultBellitem2} There exists an unbounded region $\mathcal{R} \subset (2, \infty ) \times \naturals$ such that for every pair $(p,d)\in\mathcal{R}$ there exists an increasing sequence of $n_0 < \ceil{\frac{d}{2}}$ extending subspaces $(X_n)_{n=1}^{n_0}$, where $\dim{X_n} = n$, each one orthogonal to space $\aspace$, such that $W_n = \aspace \oplus X_n$ defines a pair of quantum channels $\mathcal{N}_n, \overline{\mathcal{N}}_n  : B(H_n) \to B(K)$, $H_n\simeq W_n$, $K \simeq \complexes^d$, providing additivity breaking of minimal R\'{e}nyi \textit{p}-entropy for each $n$.
\end{enumerate}
\end{theorem}

\begin{proof}
Ad \ref{item:ResultBellitem1}. Let initially $d > 2$. By lemma \ref{lemma:ExistenceOfBellStates} one can construct $s(d)$ generalized, mutually orthonormal Bell states $\psi_i$ (we drop the two-index notation in favor of a single one for brevity) in space $\aspace^\perp$, where $\Span{\{\psi_i\}}$ is of dimension $s(d)$. Let us define space
\begin{equation}
    X_n = \Span{\{\psi_{j_k} : 1 \leqslant j_k \leqslant s(d), \, 1 \leqslant k \leqslant n\}},
\end{equation}
where $(j_k)_{k=1}^{n}$ is any subsequence of $(1, \, \ldots \, , \, s(d))$, of dimension $n$. By \eqref{eq:MuMaxSqinequality} we find the constant $A$ bounding the square of largest Schmidt coefficient of any unit vector $\xi\in W = \aspace \oplus X_n$,
\begin{equation}
    \mu_{1}^{2} \leqslant \frac{1}{2} + \sum_{k=1}^{n} \mu_{j_{k}}^{2} (\psi_{j_{k}}) = \frac{d+2n}{2d} = A.
\end{equation}
We emphasize here that condition $n < \frac{d}{2}$ must be satisfied for consistency of majorization scheme used in lemma \ref{lemma:LowerBoundValue}, otherwise one would have $A \geqslant 1$. This shows that possible dimensions of extending space $X_n$ must be less than $\ceil{\frac{d}{2}}$. Since clearly $A > \frac{1}{2}$, statement \ref{item:TheTheorem2item2} of theorem \ref{thm:TheTheorem2} yields that additivity breaking will be guaranteed by $W$ if
\begin{equation}\label{eq:Bell}
    \frac{\dim{W}}{d^2} = \frac{d(d-1)+2n}{2d^2} > \frac{(d+2n)^2}{4d^2},
\end{equation}
which comes by setting $\dim{K}=\dim{E}=d$ and $\dim{W} = \binom{d}{2}+n$. This condition is equivalent to strict positivity of a function $f_{d}$ given as
\begin{equation}
    f_{d}(n) = d(d-1)+2n - \frac{1}{2}(d+2n)^{2}.
\end{equation}
For convenience, consider an extension of $f_{d}$ on the entire interval $[1, \, \ceil{\frac{d}{2}})$ and denote it by the same symbol. First we notice $f_{d}$ is strictly monotonically decreasing (simply check that the derivative of $f_d$ is always negative). This means that if $f_{d}$ has a root $x_0 \in (1, \ceil{\frac{d}{2}})$, it is strictly positive for $x < x_0$ and strictly negative for $x > x_0$. This implies that additivity breaking will be observed for all spaces $X_n$ of dimensions $n < n_0 = \ceil{x_0}$. Calculating directly, one quickly checks that
\begin{equation}
    x_0 = \frac{1}{2}\left(1-d+\sqrt{2d^2-4d+1}\right).
\end{equation}
Furthermore, $x_0 \in (1, \ceil{\frac{d}{2}})$, i.e.~there exists at least one one-dimensional extending subspace $X_1$ leading to additivity breaking, iff $d > 6$, as claimed. Then, statement \ref{item:TheTheorem2item2} of theorem \ref{thm:TheTheorem2} guarantees the additivity breaking occurs for all $p > p_0$ for some $p_0 > 1$. On the other hand, since $n < \frac{d}{2}$ we have $\dim{W}/d^2 < \frac{1}{2}$, and so $p_0 > 2$ by \eqref{eq:Theorem2p0Bound}.

Ad \ref{item:ResultBellitem2}. This claim follows indirectly from the first one. Statement \ref{item:TheTheorem2item1} of theorem \ref{thm:TheTheorem2} yields the additivity breaking occurs for given $p$ when a function $f_{p,d}$ given as
\begin{equation}
    f_{p,d}(n) = \left[d(d-1)+2n\right]^p - \frac{1}{2^p}\left[ (d+2n)^p + (d-2n)^p \right]^{2}
\end{equation}
is strictly positive, which again comes directly from formula \eqref{eq:BreakingInequality}. As previously, consider an extension of $f_{p,d}$ on $[1, \, \ceil{\frac{d}{2}})$. This function is also strictly monotonically decreasing, since
\begin{align}
    \frac{d}{dx}f_{p,d}(x) < &-\frac{1}{2^p}\left((d+2x)^{p-1}-(d-2x)^{p-1}\right)\times \\
    &\times\left((d+2x)^p + (d-2x)^p\right) \nonumber
\end{align}
which is negative for all $x\geqslant 1$, $d\geqslant 2$ and $p > 2$. Then, additivity breaking will be observed for all spaces $X_n$ of dimensions $n < n_0 = \ceil{x_0}$ for $x_0$ being the root of $f_{p,d}$. Showing a mere existence (and an approximate location) of $x_0$ is easy. First, one checks directly that for all $p\in\reals$, $d>4$ we have
\begin{align}\label{eq:fpdBelldover4}
    f_{p,d}\left(\frac{7d}{32}\right) &= 2^{-9p} \left[ 32^p d^p (16d-9)^p - d^{2p} (9^p+23^p)^2 \right] \nonumber\\
    &< 0 ,
\end{align}
so $x_0$ exists in interval $[1, \, \ceil{\frac{d}{2}})$ whenever $f_{p,d}(1) > 0$ (we choose a point $7d/32$ as it provides a good approximation of $x_0$ and allows for easy check for \eqref{eq:fpdBelldover4}). This condition however may be checked numerically to be satisfied for all pairs $(p,d)$ in some unbounded region $\mathcal{R}$ of a real plane (see proposition \ref{prop:RegionRBell} for more concrete characterization), so an appropriate root $x_0$, and a limiting dimension $n_0$, exists for wide variety of $p$ and $d$. In order to find its approximate value one can seek for some well-behaved bounds, $\alpha_{p,d}$ and $\beta_{p,d}$ of function $f_{p,d}$, such that $\alpha_{p,d} < f_{p,d} < \beta_{p,d}$ (over interval $[1,\floor{\frac{7d}{32}}]$). For upper bound, notice $f_{p,d}(x) < \beta_{p,d}(x)$ for
\begin{equation}
    \beta_{p,d}(x) = [d(d-1)+2x]^p - \frac{1}{2^p} ( d+2x )^{2p}.
\end{equation}
For lower bound, note that $f_{p,d}$ is concave, so one could choose for $\alpha_{p,d}$ an affine function, connecting points $(1, f_{p,d}(1))$ and $(\frac{7d}{32}, f_{p,d}(\frac{7d}{32}))$. Then, the root $x_0$ will be constrained by roots $a_{p,d}$ and $b_{p,d}$ of $\alpha_{p,d}$ and $\beta_{p,d}$, respectively, i.e.~$x_0 \in [a_{p,d} ,\, b_{p,d}]$ and one finds
\begin{align}\label{eq:BoundsAB}
    a_{p,d} &= \frac{\frac{7d}{32}f_{p,d}(1)-f_{p,d}(\frac{7d}{32})}{f_{p,d}(1)-f_{p,d}(\frac{7d}{32})},\\
    b_{p,d} &= \frac{1}{2} \left( 1 - d + \sqrt{1-4d+2d^2} \right).\nonumber
\end{align}
One then checks numerically that obtained bounds for $x_0$ become quite tight for larger values of $d$ and $p$. In consequence, region $\mathcal{R}$ provides non-trivial sequences of extending spaces $X_n$. The number of possible extensions of $\aspace$ is then simply $\ceil{x_0}-1$.
\end{proof}

\begin{proposition}
\label{prop:RegionRBell}
The following statements hold:
\begin{enumerate}
    \item \label{RegRBellitem2} Theorem \ref{thm:ResultBell} fails in case $p\in(1,2]$ or $d\leqslant 6$ where no additivity breaking can be detected.
    \item \label{RegRBellitem3} For $p \in (2,\infty)$ there always exists $d_0 (p) \geqslant 7$ s.t.~for all $d \geqslant d_0 (p)$ there exists non-trivial sequence of subspaces $(X_n)$ leading to additivity breaking for given $p$.
    \item \label{RegRBellitem4} Similarly, for all $d \geqslant 7$ there exists some $p_0 (d) \in (2,\infty)$ s.t.~for all $p \in [p_0 (d),\infty)$ there exists non-trivial sequence of subspaces $(X_n)$ leading to additivity breaking for given $p$.
\end{enumerate}
\end{proposition}

\begin{proof}
Statement \ref{RegRBellitem2} follows from theorem \ref{thm:ResultBell}. For statements \ref{RegRBellitem3} and \ref{RegRBellitem4} it suffices to check for positivity of $f_{p,d}(1)$ due to monotonicity of function $n\mapsto f_{p,d}(n)$. Let us then consider a function $p\mapsto f_{p,d}(1)$ and recast it into a form
\begin{equation}
    f_{p,d}(1) = a^p - \sum_{k=1}^{3} c_k b_{k}^{p}
\end{equation}
for coefficients
\begin{align}
    a &= d(d-1)+2, \quad c_1 = c_3 = 1, \quad c_2 = 2, \\
    b_1 &= \frac{1}{2} (d-2)^2, \quad b_2 = \frac{1}{2}(d^2 - 4), \quad b_3 = \frac{1}{2}(d+2)^2.
\end{align}
Then, observe that for $d \geqslant 7$ we have $a > b_3 = \max{\{b_k\}_{k=1}^{3}}$ and $a^2 < \sum_{k=1}^{3} c_k b_{k}^{2}$ so by lemma \ref{lemma:ExpFunctionAsymptotics} (in Appendix \ref{app:Supplement}) invoked for $u=2$ we have $f_{p,d}(1) > 0$ for all
\begin{equation}
    p > p_0 (d) = \frac{\ln{4}}{\ln{(a/b_3)}} = \frac{\ln{4}}{\ln{\frac{2d(d-1)+4}{(d+2)^2}}} > 2.
\end{equation}
Conversely, consider a function $d \mapsto f_{p,d}(1)$ for some fixed $p > 2$ in a form $f_{p,d}(1) = a(d)^p - \sum_{k=1}^{3} c_k b_{k}(d)^p$, where we explicitly emphasized that coefficients $b_k = b_k (d)$ are polynomials in $d$. Condition $f_{p,d}(1) > 0$ is equivalent to
\begin{equation}
    \sum_{k=1}^{3} c_k \eta_k (d)^p < 1,
\end{equation}
where $\eta_k (d) = b_k (d) / a(d)$ are rational functions. Calculating directly, one checks that for $p > 2$ all functions $\eta_{k}^{p}$ tend to $2^{-p}$ when $d\to\infty$. By definition of a limit there will then exist constants $D_k$, $k \in \{1,2,3\}$, such that $c_1 \eta_1 (d)^{p}, c_3 \eta_3 (d)^{p} < \frac{1}{4}$ whenever $x > D_1$ or $x > D_3$, respectively, and $c_2 \eta_2 (d)^{p} < \frac{1}{2}$ when $x > D_2$. Then, for any $x > \max{\{D_k\}_{k=1}^{3}}$ we have $\sum_{k=1}^{3} c_k \eta_k (d)^p < 1$, as desired; careful analysis involving monotonicity of functions $\eta_{k}^{p}$ shows that such $\max{\{D_k\}_{k=1}^{3}}$ is then simply $D_3$, which we can easily obtain by solving inequality
\begin{equation}
    c_3 \eta_3 (d)^p = \frac{b_3 (d)^{p}}{a(d)^{p}} < \frac{1}{4},
\end{equation}
which yields a not necessarily nice formula for $d_0 (p)$,
\begin{equation}
    D_3 = d_0 (p) = \left\lceil\frac{2 + 4^{-\frac{1}{p}} + 2^{-\frac{1}{p}}\sqrt{16 - 7 \cdot 4^{-\frac{1}{p}}}}{2^{1-\frac{2}{p}}-1}\right\rceil
\end{equation}
which can be then checked to be always greater than 6, as claimed. This concludes the proof.
\end{proof}

Approximate values of $p_0 (d)$ and $d_0 (p)$, as briefly derived in the proof above, show existence of infinite set of parameters $(p,d)$ for which the additivity breaking channels can be constructed. We conducted a numerical analysis which shows that they bound the shape of region $\mathcal{R}$ quite well. As an example, in fig.~\ref{fig:RegionRBell} we present numerically obtained region $\mathcal{R}$ for $p\leqslant 10$ and $d \leqslant 75$. It may be also useful to characterize some \emph{minimal} values of parameters $p$ and $d$ giving rise to additivity breaking channels, i.e.~for which there exists at least one-dimensional extending subspace $X_n$; this is presented in a form of tables \ref{tab:MinimalD_Bell} and \ref{tab:MinimalP_Bell}.

\begin{figure}[ht]
    \centering
    \begin{tikzpicture}
    \begin{axis}[
        width=0.4\textwidth,
        height=0.45\textwidth,
        scale only axis,
        enlargelimits=false,
        axis equal=false,
        axis on top,
        xlabel={$p$},
        ylabel={$d$},
        ylabel style={rotate=-90}]
      \addplot graphics[xmin=0,xmax=10,ymin=0,ymax=75] {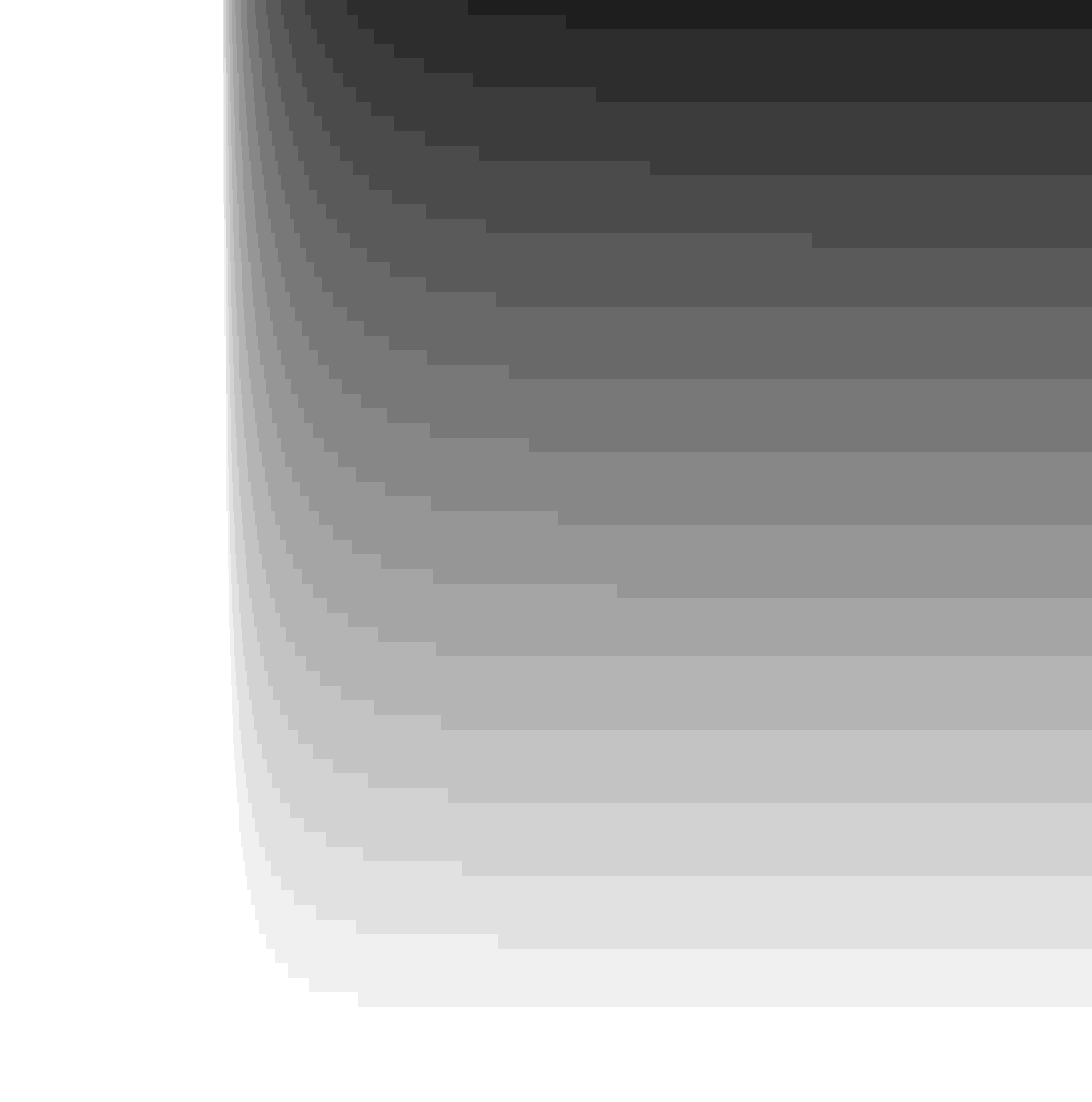};
    \end{axis}
\end{tikzpicture}
    \caption{Plot of numerically obtained region $\mathcal{R}$ appearing in a proof of theorem \ref{thm:ResultBell} for parameters $p\in [0,10]$ (horizontal axis, with increment 0.01) and $d \in [0,75]$ (vertical axis, increment 1). Each point $(p,d)$ represents a certain function $f_{p,d}$ and the shade of gray in the region encodes dimension of the largest of possible extending spaces $X_n$ determining additivity breaking channels for specified parameters $(p,d)$, equal to $\ceil{x_0}-1$ for $x_0$ the root of $f_{p,d}$. For example, for $p=4$, $d=18$ the value of $\ceil{x_0}-1$ is found to be 3, i.e.~we have $f_{18,4}(n) > 0$ for $n \in \{1,2,3\}$. In result, three possible extending spaces $X_1$, $X_2$ and $X_3$ (of dimensions 1, 2 and 3, respectively) may be constructed and each space $W_n = \aspace \oplus X_n$ ($\aspace \simeq \complexes^4 \wedge \complexes^4$) defines a pair of additivity breaking channels for $p=4$. Values of $\ceil{x_0}-1$ in the plot range from 0 (white color; no additivity breaking detection) through 1 (lightest shade) up to 15 (darkest). A ``jagged'' structure of the plot is due to a discontinuity of ceiling function. Clearly, detection of additivity breaking via theorem \ref{thm:ResultBell} fails for $p \leqslant 2$ and $d\leqslant 6$ (white color).}
    \label{fig:RegionRBell}
\end{figure}

\begin{table}[ht]
    \caption{Minimal dimension $d$ for which at least one (one-dimensional) extending subspace $X_n$, defining additivity breaking channels, may be constructed for some chosen values of $p > 2$.}
    \label{tab:MinimalD_Bell}
    \centering
    \begin{tabular}{|c|c|}\hline
        $p$           & minimal $d$ \\ \hline
        $\leqslant 2$ & no spaces   \\
        $2 + 10^{-6}$ & $>2.9 \cdot 10^6$ \\
        2.0001          & 28858        \\
        2.001          & 2890      \\
        2.01          & 293      \\
        2.1           & 33         \\
        2.2           & 19      \\
        2.5           & 11      \\
        3.0 ... 3.27 (approx.)  & 8 \\
        $\geqslant 3.28$ (approx.)  & 7 \\ \hline
    \end{tabular}
\end{table}

\begin{table}[ht]
    \caption{Minimal approximate value of $p$ ($\pm 0.001$) for which at least one (one-dimensional) extending subspace $X_n$, defining additivity breaking channels, may be constructed for some chosen values of $d > 6$.}
    \label{tab:MinimalP_Bell}
    \centering
    \begin{tabular}{|c|c|}\hline
        $d$           & minimal $p$ (approx.) \\ \hline
        $\leqslant 6$ & no spaces   \\
        7             & 3.279 \\
        8             & 2.839 \\
        9             & 2.636 \\
        10            & 2.514 \\
        15            & 2.268 \\
        20            & 2.182 \\
        100           & 2.031 \\
        500        & 2.006 \\
        1000       & 2.003 \\
        $\geqslant 2890$ & $\leqslant 2.001$ \\ \hline
    \end{tabular}
\end{table}

\subsection{Subspaces of antisymmetric space}
\label{sec:subspaces}

Here we show that appropriate space $W$ may also be chosen to be a subspace of antisymmetric space $\aspace$. Let then $W\subset\aspace$ be of dimension $\dim{W} < \binom{d}{2}$. The following technical lemma will be of importance:
\begin{lemma}\label{lemma:SchmidtOfHasSubspace}
Let $W$ be any subspace of $\aspace$ and let $P_W$ be an associated orthogonal projection. Then,
\begin{equation}
    \sup_{\|x \otimes y\| = 1}{\|P_W (x\otimes y)\|^{2}} = \frac{1}{2}.
\end{equation}
In result, the largest Schmidt coefficient $\mu_1$ of any $\xi\in W$, $\|\xi\| = 1$, satisfies $\mu_{1}^{2} \leqslant \frac{1}{2}$.
\end{lemma}

\begin{proof}
Let $P_{\mathrm{as}} : K\otimes E \to \aspace$ be the orthogonal projection onto antisymmetric space. Employing lemmas \ref{lemma:PasInequality} and \ref{lemma:SchmidtCoeffOfSubspace}, we automatically have, for any $\xi \in W$, $\|\xi\| = 1$,
\begin{align}
    \mu_{1}^{2}(\xi) &\leqslant \sup_{\|x\otimes y\|=1}{\| P_{W}(x\otimes y) \|^{2}} \\
    &\leqslant \sup_{\|x\otimes y\|=1}{\| P_{\mathrm{as}}(x\otimes y) \|^{2}} = \frac{1}{2}.\nonumber
\end{align}
We will show that in fact, the supremum $\frac{1}{2}$ is attainable. For this, choose any basis $\{w_i\}$ spanning $W$ and observe that for any such choice, there always exist some bases $\{h_i\}$, $\{f_i\}$ spanning spaces $K$ and $E$, respectively, and some bijection $\sigma : \{ 1, \, \ldots \, , \, \dim{W} \} \to \{ (\sigma_1, \sigma_2) : \sigma_1 < \sigma_2 \}$ into an ordered set of $\dim{W}$ pairs (note that the exact order used in the set of pairs is irrelevant) that every $w_i$ may be expressed as an antisymmetrized combination of simple tensors of a form $\{h_{\sigma_1} \otimes f_{\sigma_2}\}$, i.e.
\begin{equation}
    w_i = \frac{1}{\sqrt{2}} \left(h_{\sigma(i)_{1}}\otimes f_{\sigma(i)_{2}} -h_{\sigma(i)_{2}}\otimes f_{\sigma(i)_{1}} \right),
\end{equation}
where subscripts 1 and 2 indicate a specified element of a pair $\sigma(i) = (\sigma(i)_1, \sigma(i)_2)$. Expressing vectors $x$ and $y$ as appropriate ordered tuples $x = (x_i)$, $y = (y_i)$ with respect to bases $\{h_i\}$, $\{f_i\}$ we obtain
\begin{align}
    \langle x\otimes y , P_{W}(x\otimes y) \rangle &= \sum_{i=1}^{\dim{W}} |\iprod{w_i}{x\otimes y}|^{2} \\
    &= \frac{1}{2} \sum_{s} | x_{s_1}y_{s_2} - x_{s_2} y_{s_1} |^{2}, \nonumber
\end{align}
where we sum over all pairs $s = (s_1 , s_2)$ in the image of bijection $i\mapsto \sigma(i)$. Now it is enough to note that taking $x$, $y$ such that, say, $x_{s_1} = y_{s_2} = 1$, $x_{s_2} = y_{s_1} = 0$ for some specific pair $(s_1, s_2)$ and $x_{p_{1,2}} = y_{p_{1,2}} = 0$ for every other pair $(p_1,p_2)$ we have $\langle x\otimes y , P_{W}(x\otimes y) \rangle = \frac{1}{2}$. This concludes the proof.
\end{proof}

\begin{theorem}\label{thm:ResultSubspaces}
The following statements hold:
\begin{enumerate}
    \item\label{item:ResultSubspacesItem1} For every $d \geqslant 4$ there exists a finite sequence of subspaces $(W_k)_{k=1}^{N} \subset \aspace$, where
    \begin{equation}\label{eq:NumberOfSubspacesGeneral}
        N = \binom{d}{2} - 1 - \floor{4^{-1}d^2},
    \end{equation}
    and
    \begin{equation}
        \dim{W_{N}} = \binom{d}{2} - 1, \quad \dim{W_{k-1}} = \dim{W_k} - 1,
    \end{equation}
    such that each subspace $W_k$ defines a pair of quantum channels $\mathcal{N}_k, \overline{\mathcal{N}}_k  : B(H_k) \to B(K)$, $H_k\simeq W_k$, $K \simeq \complexes^d$, providing additivity breaking of minimal R\'{e}nyi \textit{p}-entropy for all $p\in (p_0,\infty)$ for some $p_0 > 2$.
    \item\label{item:ResultSubspacesItem2} For all $p>2$ there exists a dimension $d_0$ such that for all dimensions $d>d_0$ there exists an increasing sequence of subspaces $(W_k)$ in antisymmetric space $\aspace \simeq \complexes^d \wedge \complexes^d$ of length $l_{p,d}$ given via formula
\begin{equation}\label{eq:NumberOfSubspaces}
    l_{p,d} = \binom{d}{2}-1-\floor{4^{\frac{1}{p}-1}d^2},
\end{equation}
such that every subspace $W_k$ defines a pair of quantum channels $\mathcal{N}_k$, $\overline{\mathcal{N}}_k : B(H_k) \to B(K)$, $H_k \simeq W_k$, $K \simeq \complexes^d$, providing additivity breaking of minimal R\'{e}nyi \textit{p}-entropy.
\end{enumerate}
\end{theorem}

\begin{proof}
Ad \ref{item:ResultSubspacesItem1}. Let $W \subset \aspace$ be any subspace of dimension $\dim{W} < \binom{d}{2}$. Lemma \ref{lemma:SchmidtOfHasSubspace} yields a square of the largest Schmidt coefficient of any unit vector in $W$ is bounded by $A = \frac{1}{2}$. Statement \ref{item:TheTheorem2item2} of theorem \ref{thm:TheTheorem2} then implies that additivity breaking occurs if
\begin{equation}\label{eq:Subspaces}
    \frac{\dim{W}}{d^2} > \frac{1}{4}.
\end{equation}
From this we conclude that if there exist some natural numbers $(n_k)_{k=1}^{N}$ in interval $\mathcal{I}_d = \left(\frac{1}{4}d^2 , \binom{d}{2}\right)$ there will exist additivity breaking subspaces $(W_{k})_{k=1}^{N}$ in $\aspace$, each of dimension $\dim{W_k} = n_k$. This is the case iff $\mathcal{I}_d$ is of length greater than 1,
\begin{equation}
    \frac{1}{2}d(d-1) - \frac{1}{4}d^2 > 1,
\end{equation}
i.e.~contains at least one natural number $n_1 = n_N = \binom{d}{2}-1$, which is true for all $d \geqslant 4$. A total number $N$ of such subspaces equals a number of all natural numbers in $\mathcal{I}_d$, which can be quickly checked to be indeed given by formula \eqref{eq:NumberOfSubspacesGeneral}. Possible dimensions of subspaces $W_k$ are then
\begin{align}
    &n_1 = 1+\floor{4^{-1}d^2}, \quad n_2 = 2+\floor{4^{-1}d^2}, \\
    &... \quad n_N = \binom{d}{2}-1.\nonumber
\end{align}
For any subspace $W\subset\aspace$ and for any $d \geqslant 4$ we have $\dim{W}/d^{2} < \frac{1}{2}$. Therefore, \eqref{eq:Theorem2p0Bound} yields $p_0 > 2$.

Ad \ref{item:ResultSubspacesItem2}. Similarly to theorem \ref{thm:ResultBell}, this statement is a rephrasing of previous one. Theorem \ref{thm:TheTheorem2}, p. \ref{item:TheTheorem2item1} yields the additivity breaking condition is equivalent to strict positivity of a function
\begin{equation}
    f_{p,d}(n) = \frac{n^p}{d^{2p}} - 4^{1-p}.
\end{equation}
For given $p>2$ this yields $n > 4^{\frac{1}{p}-1}d^2$. As previously, since we seek for $n$-dimensional subspaces of the space of dimension $\binom{d}{2}$, we have to ask if there exist any natural numbers $n < \binom{d}{2}$ such that the inequality $n > 4^{\frac{1}{p}-1}d^2$ is satisfied, i.e.~if set $\naturals \cap \big(4^{\frac{1}{p}-1}d^2 , \binom{d}{2}\big)$ is nonempty. This will be guaranteed if the interval $\mathcal{I}_{p,d} = \big(4^{\frac{1}{p}-1} d^2 , \binom{d}{2}\big)$ is of a length greater than 1, i.e.
\begin{equation}
    \frac{d(d-1)}{2} - 4^{\frac{1}{p}-1} d^2 > 1,
\end{equation}
which is satisfied by all dimensions $d > d_0$ for $d_0$ given via formula
\begin{equation}\label{eq:d0pSubspaces}
    d_0 = \left\lceil 4\left(-1 + \sqrt{9-4^{\frac{1}{p}+1}}\right)^{-1} \right\rceil .
\end{equation}
In consequence, for all $p>2$ and all $d>d_0$ there will exist a finite sequence of natural numbers $(n_k)_{k=1}^{N}$ such that $f_{p,d}(n_k) > 0$, i.e.~there will exist a sequence of additivity breaking subspaces $(W_k)_{k=1}^{N}$, $\dim{W_k} = n_k$. Number of such spaces, i.e.~a length $l_{p,d}$ of this sequence is equal to total number of all natural numbers $n_k$ in the interval $\mathcal{I}_{p,d}$; this however is easily seen to be given by claimed equation \eqref{eq:NumberOfSubspaces}. Moreover, monotonicity of function $n \mapsto f_{p,d}(n)$ yields that in each case a dimension of the largest subspace $W_k$ in $\aspace$, is always found to be $\max{\{n_k\}}=\binom{d}{2} - 1$.
\end{proof}

\begin{proposition}
\label{prop:RegionRSubspaces}
The following statements hold:
\begin{enumerate}
    \item \label{RegRSubspacesitem1} The additivity breaking criterion in theorem \ref{thm:ResultSubspaces} is inconclusive in case $p \in (1,2]$ or $d \leqslant 3$, where no additivity breaking can be detected.
    \item \label{RegRSubspacesitem3} For $p\in(2,\infty)$ there always exists $d_0 (p) \geqslant 4$ s.t.~for all $d \geqslant d_0 (p)$ there exists a non-trivial sequence of subspaces $(W_k)$ leading to additivity breaking for given $p$.
    \item \label{RegRSubspacesitem4} Similarly, for all $d \geqslant 4$ there exists some $p_0 (d) \in (2,\infty)$ s.t.~for all $p \in [p_0 (d),\infty)$ there exists a non-trivial sequence of subspaces $(W_k)$ leading to additivity breaking for given $p$.
\end{enumerate}
\end{proposition}

\begin{proof}
Again, statement \ref{RegRSubspacesitem1} follows from theorem \ref{thm:ResultSubspaces}. The estimation for initial dimension $d_0 (p)$ appearing in statement \ref{RegRSubspacesitem3} can be expressed in form of formula \eqref{eq:d0pSubspaces} already calculated in a proof of theorem \ref{thm:ResultSubspaces}. For the remaining statement \ref{RegRSubspacesitem4} it is enough, due to monotonicity of function $n\mapsto f_{p,d}(n)$, to simply consider positivity of $f_{p,d}(\dim{\aspace}-1)$,
\begin{equation}
    f_{p,d}\left(\binom{d}{2}-1\right) = \left(\frac{d(d-1)-2}{2d^2}\right)^p - 4^{1-p} > 0,
\end{equation}
which holds for all $p > p_0 (d)$ such that
\begin{equation}
    p_0 (d) = \frac{\ln{4}}{\ln{[2d(d-1)-4]} - 2 \ln{d}} > 2, \quad d \geqslant 4.
\end{equation}
This concludes the proof.
\end{proof}

\begin{figure}
    \centering
    \begin{tikzpicture}
    \begin{axis}[
        width=0.4\textwidth,
        height=0.45\textwidth,
        scale only axis,
        enlargelimits=false,
        axis equal=false,
        axis on top,
        xlabel={$p$},
        ylabel={$d$},
        ylabel style={rotate=-90}]
      \addplot graphics[xmin=2,xmax=10,ymin=2,ymax=15.5] {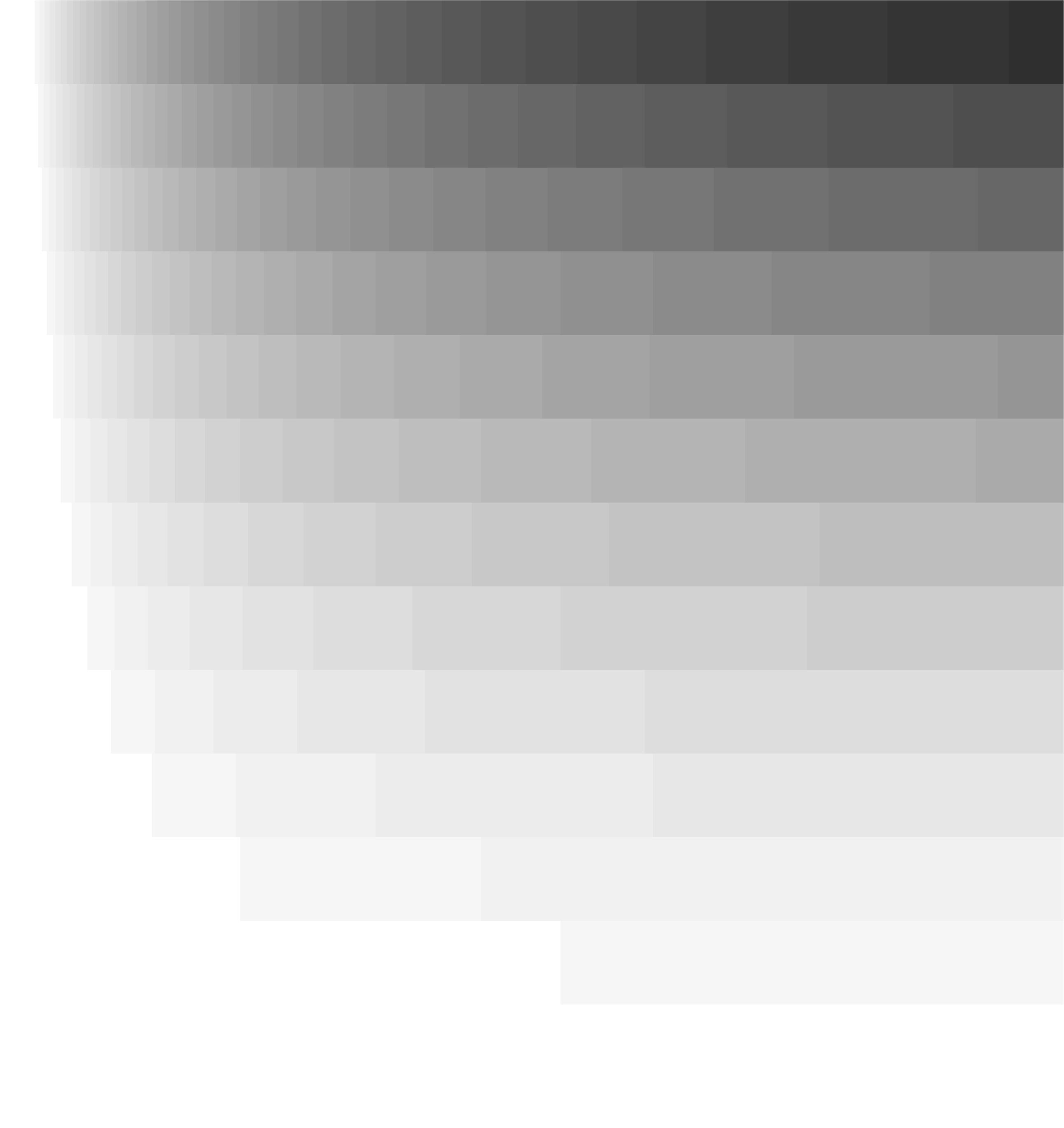};
    \end{axis}
\end{tikzpicture}
    \caption{Plot of numerically obtained region $\mathcal{R}$ appearing in a proof of theorem \ref{thm:ResultSubspaces} for parameters $p\in [2,10]$ (horizontal axis, with increment 0.001) and $d \in [2,15]$ (vertical axis, increment 1). Each point $(p,d)$ represents a certain function $f_{p,d}$ and the shade of gray in the region encodes value of $l_{p,d}$ equal to a number of possible subspaces $W_k \subset \aspace$ determining additivity breaking channels for specified parameters $(p,d)$. Values of $l_{p,d}$ in the plot range from 0 (white color; no additivity breaking detection) through 1 (lightest shade) up to 41 (darkest). A ``jagged'' structure of the plot is due to a discontinuity of ceiling function. Detection of additivity breaking via theorem \ref{thm:ResultSubspaces} fails for $p \in (2,\infty)$ and $d\leqslant 3$ (white color).}
    \label{fig:RegionRSubspaces}
\end{figure}

\begin{table}[ht]
    \caption{Minimal dimension $d$ for which at least one subspace $W_k \subset \aspace$ (of dimension $\binom{d}{2}-1$), defining additivity breaking channels, may be constructed for some chosen values of $p > 2$.}
    \label{tab:MinimalD_Subspaces}
    \centering
    \begin{tabular}{|c|c|}\hline
        $p$           & minimal $d$ \\ \hline
        $\leqslant 2$ & no spaces   \\
        $2 + 10^{-6}$ & $>2.8 \cdot 10^6$ \\
        2.0001          & 28858        \\
        2.001          & 2890      \\
        2.01          & 293      \\
        2.1           & 33         \\
        2.2           & 19      \\
        2.5           & 10      \\
        2.6           & 9 \\
        2.7, 2.8      & 8 \\
        2.9 ... 3.1 & 7 \\
        3.2 ... 3.8 & 6 \\
        3.9 ... 6.2 & 5 \\
        $\geqslant 6.3$ (approx.) & 4 \\ \hline
    \end{tabular}
\end{table}

\begin{table}[ht]
    \caption{Minimal approximate value of $p$ ($\pm 0.001$) for which at least one subspace $W_k \subset \aspace$ (of dimension $\binom{d}{2}-1$), defining additivity breaking channels, may be constructed for some chosen values of $d > 3$.}
    \label{tab:MinimalP_Subspaces}
    \centering
    \begin{tabular}{|c|c|}\hline
        $d$           & minimal $p$ (approx.) \\ \hline
        $\leqslant 3$ & no spaces   \\
        4             & 6.213 \\
        5             & 3.802 \\
        6             & 3.138 \\
        7             & 2.828 \\
        8             & 2.650 \\
        9             & 2.534 \\
        10            & 2.453 \\
        15            & 2.256 \\
        20            & 2.178 \\
        100           & 2.031 \\
        200           & 2.015 \\
        500           & 2.006 \\
        1000       & 2.003 \\
        $\geqslant 2890$ & $\leqslant 2.001$ \\ \hline
    \end{tabular}
\end{table}

 \section{Conclusions and open problems}
In this paper we address the problem of finding more examples of subspaces for which quantum channels generated by them exhibit property of additivity breaking of the minimum output Rényi \textit{p}-entropy. To show that considered spaces generate interesting in our context quantum channel one has to resolve the main technical obstacle which is derivation of respective bounds on the maximal Schmidt coefficient. It turns out, such bound leads us to bound on minimum output \textit{p}-entropy for two copies of the channel. Finding a bound on the maximal Schmidt coefficient is however, in general very hard and complex problem and probably it is one of the main reason why number of examples of quantum channels with considered property is so small. 
However, we were able present the explicit construction for three different cases: extensions of the antisymmetric space and subspaces of the antisymmetric space. We present the full analytical analysis showing channels for which we have additivity breaking for every $p>2$ and dimension $d$ large enough.

We would like to stress here that we are interested in possibly tight bounds on considered Schmidt coefficients, to ensure as low dimensionality of the resulting channels as possible. This of course, at least potentially leaves a huge area for possible contribution by considering other spaces than antisymmetric one. An interesting problem would be to consider subspaces structurally very different for antisymmetric but still having some operative properties, allowing for certain derivations crucial for the topic of this manuscript. One could for example consider spaces coming from group/algebra theoretical considerations. Such additional symmetries, in principle, could lead us to easier analytical work and possibly result in analytical solutions. Good candidates are subspaces defined through the famous Schur-Weyl duality \cite{harrisBook,FultonSchur} and their deformations by partial transposition \cite{Mozrzymas_2018}. Especially, in the latter case, one could expect interesting results since the characterization of these spaces is recent and we do not have yet any checks in this regard. However, this is still a very complicated technical problem and we leave it for possible further research.

\appendix

\section{Mathematical supplement}
\label{app:Supplement}

\subsection{Proof of Lemma \ref{lemma:UpperBoundValue}}
\label{app:lemmaUpperBoundValueProof}

Let us define $\rho\in B(K_1 \otimes K_2)$ as $\rho = \ptr{E_1 E_2}{\proj{\eta(\psi^+)}}$ and introduce Schmidt decomposition of $\eta (\psi^+)$, i.e.~in cut $K_1 K_2 : E_1 E_2$, as $\eta(\psi^+) = \sum_{i} \mu_{i} u_i \otimes v_i$ for orthonormal bases $\{u_i\}$ and $\{v_i\}$ spanning $K_1 \otimes K_2$ and $E_1 \otimes E_2$, respectively, and coefficients $\mu_i \geqslant 0$; denote also $\mu_1 = \max{\{\mu_i\}}$. Then we again have $\rho = \sum_i \mu_{i}^{2} \proj{u_i}$ and finding the lower bound on largest Schmidt coefficient of $\eta(\psi^+)$ is equivalent to finding bound on largest eigenvalue of $\rho$, i.e.~its spectral radius, $\mu_{1}^{2} = \sup_{\| x \| = 1}{\langle x, \rho (x) \rangle}$. In particular, the following inequality holds,
\begin{equation}
    \mu_{1}^{2} = \sup_{\| x \| = 1}{\langle x, \rho (x) \rangle} \geqslant \langle \phi^+ , \rho (\phi^+) \rangle = \ptr{K_1 K_2}{P^+\rho},
\end{equation}
where $\phi^+$ is a maximally entangled state between spaces $K_1$, $K_2$ and $P^+ = \proj{\phi^+}$ is a associated rank one projection. Notice that since the only effect of bijection $\eta$ is a rearrangement of vectors in tensor product, we clearly have $\ptr{E_1 E_2}{\proj{\eta(\psi^+)}} = \ptr{E_1 E_2}{\proj{\psi^+}}$ and so
\begin{align}
    \mu_{1}^{2} &\geqslant \ptr{K_1 K_2}{P^+ \rho} = \ptr{K_1 K_2}{\left[ P^+ \left(\ptr{E_1 E_2}{\proj{\psi^+}}\right)\right]} \\
    &= \tr{\left[ \left( P^+ \otimes \id_{E_1 E_2} \right) \proj{\psi^+}\right]}.\nonumber
\end{align}
Let $Q^+$ be a rank one projection onto a maximally entangled state between spaces $E_1$, $E_2$. Since its orthogonal complement $\id_{E_1 E_2} - Q^+$ is also a projection, it is positive semi-definite and so we have $\id_{E_1 E_2} \geqslant Q^+$. This yields, after simple algebra, to
\begin{align}
    \mu_{1}^{2} &\geqslant \tr{\left[ \left( P^+ \otimes \id_{E_1 E_2} \right) \proj{\psi^+}\right]}\\
    &\geqslant \tr{\left[ \left( P^+ \otimes Q^+ \right) \proj{\psi^+}\right]}.\nonumber
\end{align}
Immediately, we notice $P^+ \otimes Q^+ = \proj{\Psi^+}$, where $\Psi^+$ stands for a maximally entangled state between spaces $K_1 \otimes E_1$ and $K_2 \otimes E_2$. This leads to
\begin{equation}\label{eq:mu1lowerBound}
    \mu_{1}^{2} \geqslant \tr{\left[ \left( P^+ \otimes Q^+ \right) \proj{\psi^+}\right]} = \left| \langle \psi^+ , \Psi^+ \rangle \right|^{2}.
\end{equation}
For any possible choice of $W$ and its basis $\{w_i\}$, one can always select a basis $\{b_i\}$ in $K\otimes E$ in such a way that $\{w_i\}$ is a subset of $\{b_i\}$. Then, all the remaining vectors $b_i \notin W$ constitute for a basis in orthogonal complement $W^\perp$; denote them by $w_{i}^{\perp}$. Let 
\begin{equation}
    \theta = \sum_{i=1}^{\dim{W}} w_i \otimes w_i , \quad \theta^\perp = \sum_{i=1}^{\dim{W^{\perp}}} w_{i}^{\perp} \otimes w_{i}^{\perp},
\end{equation}
with $\theta^\perp$ clearly orthogonal to $\theta$. Then, the two maximally entangled states $\psi^+$ and $\Psi^+$ may be expressed as
\begin{equation}
    \psi^+ = \frac{\theta}{\sqrt{\dim{W}}} , \quad \Psi^+ = \frac{\theta + \theta^\perp}{\sqrt{\dim{K}\dim{E}}}.
\end{equation}
By direct check, the inner product in \eqref{eq:mu1lowerBound} is then simply proportional to $\| \theta \|^{2} = \dim{W}$. We immediately obtain
\begin{align}
     \mu_{1}^{2} &\geqslant \left| \langle \psi^+ , \Psi^+ \rangle \right|^{2} = \frac{\| \theta \|^{4}}{\dim{W}\dim{K}\dim{E}} \\
     &= \frac{\dim{W}}{\dim{K}\dim{E}}.\nonumber
\end{align}

\subsection{Some secondary lemmas and proofs}

\begin{lemma}\label{lemma:ExpFunctionAsymptotics}
    Let $a$, $b_1$, ... , $b_n$, $c_1$, ... , $c_n$ be positive numbers such that $a > \max{\{b_k\}_{k=1}^{n}}$ and let
    \begin{equation}
        g(x) = a^x - \sum_{k=1}^{n} c_k b_{k}^{x} .
    \end{equation}
    If $g(u) < 0$ for some $u > 0$ then $g$ admits a single root $x_0 > u$ and $\left.g\right|_{(x_0,\infty)}$ is strictly increasing.
\end{lemma}

\begin{proof}
The proof follows almost directly from asymptotic behavior of exponential functions. From $a > b_k$ we have
\begin{align}
    &\sum_{k=1}^{n} c_k b_{k}^{x} < C b_{0}^{x} \\
    &\text{for}\quad b_0 = \max{\{b_k\}_{k=1}^{n}}, \quad C = \sum_{k=1}^{n} c_k .\nonumber
\end{align}
First, if $C < 1$, we have $\left(\frac{a}{b_0}\right)^x > C$ for all $x > 0$ and $g(x)$ is positive everywhere on $\reals_+$. Second, if $C > 1$, we have $\left(\frac{a}{b_0}\right)^x > C$, and $g(x) > 0$ in consequence, for all $x > \ln{C} / \ln{\frac{a}{b_0}}$. This, together with continuity of $g$, proves the existence of a root $x_0$ in interval $(u, \ln{C} / \ln{\frac{a}{b_0}})$. Let then $g(x_0) = 0$. This gives
\begin{equation}\label{eq:aTox0}
    a^{x_0} = \sum_{k=1}^{n} c_k b_{k}^{x_0} .
\end{equation}
The derivative of $g$ in $x_0$ is, employing \eqref{eq:aTox0},
\begin{align}
    \frac{dg}{dx}(x_0) &= a^{x_{0}} \ln{a} - \sum_{k=1}^{n} c_k b_{k}^{x_0} \ln{b_k} \\
    &= \sum_{k=1}^{n} c_k b_{k}^{x_0} \ln{\frac{a}{b_k}}\nonumber
\end{align}
which is however strictly positive, i.e.~$g$, being continuously differentiable, is monotonically increasing in some open neighborhood of $x_0$. Since this observation applies to any possible root $x_0$, we conclude there exists only one root. Now, consider point $x_0 + \delta$, $\delta > 0$ and observe
\begin{equation}
    \frac{dg}{dx}(x_0 + \delta) = \sum_{k=1}^{n} c_k b_{k}^{x_0} \left( a^\delta \ln{a} - b_{k}^{\delta}\ln{b_k} \right)
\end{equation}
which is again checked to be strictly positive with some simple algebra involving monotonicity of logarithm and exponential function. This shows $g$ is indeed increasing over interval $(x_0, \infty)$.
\end{proof}

\subsection{Majorization and Schur concavity}
\label{appMajorization}
Here we give a brief note on the relation of \emph{majorization} in finite-dimensional real spaces, which plays a substantial role in a proof of Lemma \ref{lemma:LowerBoundValue}. Let $x = (x_i)\in \reals^n$ and denote by $x_\downarrow = (x_{i}^{\downarrow})$ a vector of components of $x$ in a \emph{non-increasing order}, i.e.~$x_{1}^{\downarrow} \geqslant x_{2}^{\downarrow} \geqslant ... \geqslant x_{n}^{\downarrow}$. Then, we say that vector $y \in \reals^n$ \emph{majorizes} vector $x$, or that $x \preceq y$, if
\begin{equation}
    \sum_{i=1}^{k} x_{i}^{\downarrow} \leqslant \sum_{i=1}^{k} y_{i}^{\downarrow} \quad \text{for all } k\in \{ 1, \, \ldots \, , \, n \}.
\end{equation}
Such relation can be then shown to be a \emph{preorder} on $\reals^n$, which however fails to be a partial order.
\vskip\baselineskip
A function $f : \reals^n \to \reals$ is said to be a \emph{Schur convex function}, if for any vectors $x,y\in\reals^n$ such that $x \preceq y$ we have $f(x) \leqslant f(y)$. If, on the hand, $x\preceq y$ implies $f(x) \geqslant f(y)$, function $f$ is called \emph{Schur concave}.

\subsection{Schmidt decomposition and some bounds on its largest coefficient}
\label{appA}

Let $H_1$, $H_2$ be finite-dimensional Hilbert spaces. Then, it is true that for every vector $h \in H_1 \otimes H_2$ there exist a sequence $(\mu_i) \in \reals_{+}^{n}$ and orthonormal systems $(u_i)_{i=1}^{n} \subset H_1$, $(v_i)_{i=1}^{n} \subset H_2$ for $n \leqslant \min{\{\dim{H_1}, \dim{H_2}\}}$ such that $h$ admits the so-called \emph{Schmidt decomposition}
\begin{equation}
    h = \sum_{i=1}^{n} \mu_i u_i \otimes v_i .
\end{equation}
Then, numbers $\mu_i \geqslant 0$ are sometimes called the \emph{Schmidt coefficients}. In all the following, notation $\mu_1 (\xi)$ will denote the largest Schmidt coefficient of some vector $\xi$ (usually of unit length).

\begin{lemma}\label{lemma:SchmidtMaxAsSupremum}
    Let $(H_1, \iprod{\cdot}{\cdot}_1)$, $(H_2, \iprod{\cdot}{\cdot}_2)$ be Hilbert spaces and let $h \in H_1 \otimes H_2$, $\| h \| = 1$. Then, the largest Schmidt coefficient $\mu_1 (h)$ can be computed via the formula
    \begin{equation}
        \mu_{1}^{2}(h) = \sup_{\|x\otimes y\|=1}{|\iprod{h}{x\otimes y}|^2}.
    \end{equation}
\end{lemma}

\begin{proof}
    Let $h \in H_1 \otimes H_2$ be given via its Schmidt decomposition
    \begin{equation}
        h = \sum_i \mu_i u_i \otimes v_i
    \end{equation}
    for some orthonormal systems $\{u_i\}\subset H_1$, $\{v_i\}\subset H_2$ and $\mu_i \geqslant 0$. Let $\mu_1 = \max{\{\mu_i\}}$. For $\|x \otimes y\| = 1 $ we have
    \begin{align}
        |\iprod{h}{x\otimes y}|^2 &= \left| \sum_{i} \mu_i \iprod{u_i}{x}_1 \iprod{v_i}{y}_2 \right|^2 \\
        &\leqslant \left(\sum_i \mu_i |\iprod{u_i}{x}_1| |\iprod{v_i}{y}_2| \right)^2 \nonumber \\
        &\leqslant \mu_{1}^{2} \sum_i |\iprod{u_i}{x}_1|^{2} \sum_j |\iprod{v_i}{y}_2|^2 \nonumber \\
        &\leqslant \mu_{1}^{2} \| x \|^2 \|y \|^2 \nonumber = \mu_{1}^{2}, \nonumber
    \end{align}
    where we employed H\"{o}lder's and Bessel's inequalities. Now it suffices to take $x = u_1$, $y = v_1$ to check that the upper bound $\mu_{1}^{2}$ is attainable.
\end{proof}

\begin{lemma}\label{lemma:SchmidtCoeffOfProjection}
    Let $X$ be a subspace in Hilbert space $H$ and let $P : H \to X$ be the associated orthogonal projection. Then, for every vector $\xi \in X$, $\| \xi \| = 1$, we have
    \begin{equation}
        \mu_{1}^{2}(\xi) \leqslant \sup_{\|x\otimes y\| = 1}{\|P (x\otimes y)\|^{2}}.
    \end{equation}
\end{lemma}

\begin{proof}
    Let $\{e_i\}$ be an orthonormal basis spanning $X$, so $P = \sum_{i} \proj{e_i}$. Let $\xi \in X$, $\| \xi \| = 1$, be given as $\xi = \sum_i c_i e_i$ for coefficients satisfying $\sum_i |c_i|^2 = 1$. By lemma \ref{lemma:SchmidtMaxAsSupremum} we have
    \begin{align}
        \mu_{1}^{2}(\xi) &= \sup_{\|x\otimes y\|=1}{|\iprod{\xi}{x\otimes y}|^2} \\
        &= \sup_{\|x\otimes y\|=1}{\left|\sum_i \overline{c_i}\iprod{e_i}{x\otimes y}\right|^2} \nonumber\\
        &\leqslant \sup_{\|x\otimes y\|=1}{\left(\sum_i |c_i||\iprod{e_i}{x\otimes y}|\right)^2} \nonumber\\
        &\leqslant \sup_{\|x\otimes y\|=1}{\sum_i |c_i|^2 \sum_j |\iprod{e_i}{x\otimes y}|^2} \nonumber\\
        &= \sup_{\|x\otimes y\|=1}{\sum_i |\iprod{e_i}{x\otimes y}|^2} \nonumber \\
        &= \sup_{\|x\otimes y\|=1}{\left\|\sum_i \iprod{e_i}{x\otimes y}e_i\right\|^2} \nonumber\\
        &= \sup_{\|x\otimes y\|=1}{\| P(x\otimes y)\|^{2}},\nonumber
    \end{align}
    which comes by H\"{o}lder's inequality and orthogonality of basis $\{e_i\}$.
\end{proof}

\begin{lemma}\label{lemma:SchmidtCoeffOfDirectSum}
    Let $X$ be a subspace in Hilbert space $H$ such that $X = X_1 \oplus X_2$ for some mutually orthogonal subspaces $X_1$, $X_2$ and denote by $P$, $P_1$ and $P_2$ orthogonal projections onto $X$, $X_1$ and $X_2$, respectively. Then, for all unit vectors $\xi \in X$ we have
    \begin{align}
        \mu_{1}^{2}(\xi) &\leqslant \sup_{\|x\otimes y\|=1}{\|P (x\otimes y)\|^{2}} \\
        &\leqslant \sup_{\|x\otimes y\|=1}{\|P_1 (x\otimes y)\|^{2}} + \sup_{\|x\otimes y\|=1}{\|P_2 (x\otimes y)\|^{2}}.\nonumber
    \end{align}
\end{lemma}

\begin{proof}
    Take any $\xi \in H$, $\|\xi\| = 1$. The left-most inequality comes directly from lemma \ref{lemma:SchmidtCoeffOfProjection}. For the remaining one, put $P = P_1 + P_2$ and use orthogonality of $\ran{P_1}$ and $\ran{P_2}$ so $\| P (x\otimes y) \|^2 = \| P_1 (x\otimes y) \|^2 + \| P_2 (x\otimes y) \|^2$. The result then comes after utilizing property $\sup_{x}{\{f(x)+g(x)\}} \leqslant \sup_{x}{f(x)} + \sup_{x}{g(x)}$.
\end{proof}

\begin{lemma}\label{lemma:SchmidtCoeffOfSubspace}
    Let $X$, $Y$ be subspaces in Hilbert space $H$ such that $X\subset Y\subset H$ and denote by $P_X$ and $P_Y$ respective orthogonal projections. Then, for all vectors $\xi \in X$, $\|\xi \| = 1$, we have
    \begin{align}
        \mu_{1}^{2}(\xi) &\leqslant \sup_{\|x\otimes y\| = 1}{\|P_X (x\otimes y)\|^{2}} \\
        &\leqslant \sup_{\|x\otimes y\| = 1}{\|P_Y (x\otimes y)\|^{2}}.\nonumber
    \end{align}
\end{lemma}

\begin{proof}
    Again, the left-most inequality is satisfied by lemma \ref{lemma:SchmidtCoeffOfProjection}. For the second one, notice that $Y = X \oplus X^\perp$ for $X^\perp$ being orthogonal to $X$ and denote by $P^\perp$ the corresponding projection; then $P^\perp = P_Y - P_X$ and it is a positive semi-definite operator, i.e.
    \begin{align}
        0 &\leqslant \iprod{x\otimes y}{(P_Y-P_X)(x\otimes y)} \\
        &= \|P_Y(x\otimes y)\|^{2} - \|P_X(x\otimes y)\|^{2} \nonumber
    \end{align}
and the claim follows.
\end{proof}

\section*{Acknowledgments}
The authors are indebted to Professor Micha{\l} Horodecki for discussion and to the anonymous Referee for his outstanding, throughout comments which allowed for significant improvements of the early version of the manuscript. M.S.~acknowledges support by grant Sonata 16, UMO-2020/39/D/ST2/01234 from the Polish National Science Centre. For the purpose of Open Access, the author has applied a CC-BY public copyright license to any Author Accepted Manuscript (AAM) version arising from this submission.

\bibliographystyle{unsrt}
\bibliography{bibliography.bib}

\end{document}